\documentclass[12pt, a4paper]{article}
\usepackage[utf8]{inputenc}
\usepackage[margin=0.75in]{geometry}
\usepackage{booktabs, float}
\usepackage{tabularx}
\usepackage{hyperref}
\hypersetup{
	colorlinks=true,
	linkcolor=blue,
	citecolor=blue,
    urlcolor=blue
}
\usepackage{amsmath, amsfonts, amssymb, amsthm}
\usepackage{enumerate}

\newtheorem{theorem}{Theorem}
\newtheorem{lemma}{Lemma}
\usepackage{algorithm}
\usepackage{algpseudocode}
\usepackage{appendix}
\usepackage{amssymb,amsmath,relsize}
\usepackage{natbib}
\usepackage{graphicx} 
\usepackage{subfigure}
\usepackage{rotating}
\usepackage{adjustbox}
\usepackage{csquotes}
\usepackage{mathrsfs}
\usepackage{isomath}
\usepackage{lscape}
\usepackage{multirow}
\usepackage{tikz}
\usetikzlibrary{arrows.meta, positioning}
\DeclareUnicodeCharacter{2212}{-}
\begin{document}
	
	\title{Optimal Designs in Multicomponent Stress–Strength Reliability for the Unit Generalized Rayleigh Distribution}
	\author{Rajat Das$^1$, Yogesh Mani Tripathi$^{1}$, Tanmay Kayal$^2$\thanks{Corresponding author:tanmay@iitmandi.ac.in}
	}
	\date{}
	{\footnotesize  \maketitle \noindent {\it $^{1}$Department of Mathematics, Indian Institute of Technology Patna, Patna, Bihar-801106, India}\\
    {\it $^2$School of Mathematical \& Statistical Sciences, Indian Institute of Technology Mandi, Kamand, 175075, Himachal Pradesh, India }} 
	\maketitle
	\begin{abstract}
		\noindent  
        A unified inferential framework is developed to address the stress-strength reliability of multicomponent systems under progressive Type II censoring. The maximum likelihood estimate of reliability is obtained using an expectation-maximization algorithm, followed by the determination of the corresponding Fisher information matrix and confidence intervals based on the missing-value principle. To facilitate a comparative inferential assessment, maximum product spacing estimates are also developed. By employing both informative and non-informative prior models, a comprehensive analysis is conducted within a Bayesian framework, and suitable summaries are obtained using the Markov chain Monte Carlo algorithm. The performance of all the estimators is analyzed through an extensive simulation study. Finally, a practical application of the proposed methodology is presented using a reliability data set. Furthermore, we determine optimal progressive censoring strategies using three different optimality measures and discuss their usefulness in reliability studies.

		~~\\ 
		\noindent {\it Keywords:}~ Unit Generalized Rayleigh Distribution, Multicomponent reliability, EM algorithm, Maximum product of spacing, Bayesian inference, Optimality.\\

        \noindent {\it AMS Classifications:}~ 62F10, 62F15, 62N02.
	\end{abstract}
	
\section{Introduction}
The stress-strength reliability model, first introduced by \citet{Church}, is one of the preeminent concepts in reliability engineering and survival analysis. It provides a probabilistic measure of component performance by quantifying the likelihood that the inherent strength of a unit exceeds the external stress applied to it during operation. If $X$ represents the random strength and $Y$ denotes the random stress, the reliability is defined as $\Psi = P(X > Y)$, which indicates the probability that the component functions safely without failure. This framework is widely used across mechanical engineering, materials science, aerospace systems, electronic devices, and quality control, where stress values may arise from loads, temperatures, voltages, or pressures, and strength represents the component's resistance or endurance. Because failure occurs whenever the applied stress surpasses the strength, the stress–strength model offers a direct and intuitive tool for assessing durability and safety margins. Extensive research has been conducted on estimating $\Psi$ under various lifetime distributions and data structures, including \citet{Weerahandi}, \citet{kundu}, \citet{ghitany}, \citet{cetinkaya}, \citet{Wang_csda}, \citet{sarhan}, \citet{sharma}, \citet{asadi}, \citet{surles}, \citet{Krishnamoorthy}. Real-life testing often involves censored or incomplete observations, motivating the study of stress–strength reliability under Type I, Type II, and progressive censoring schemes. These developments enable more realistic modeling of experimental constraints and practical testing situations. As a result, the stress–strength model has evolved into a versatile analytical framework, forming the basis for reliability evaluation in both single-component and multicomponent systems, and continues to play a vital role in modern reliability analysis.

The $s$-out-of-$v$ system is a widely used reliability structure that models partially redundant engineering systems. In this configuration, a system composed of $v$ identical components continues to operate successfully as long as at least $s$ of those components remain functional. Thus, system failure occurs only when $v-s+1$ components have failed. This framework generalizes classical reliability systems: when $s = 1$, it reduces to a parallel system, and when $s = v$, it becomes a series system. The flexibility of this structure allows it to accurately model diverse real-world applications such as communication networks, power grids, manufacturing processes, aircraft control systems, and mechanical assemblies, where partial capacity is sufficient for proper operation. Mathematically, if $X_1, X_2, \ldots, X_v$ denote the lifetimes of the components, the system lifetime is given by the $(v-s+1)$-th order statistic $T_{s:v} = X_{v-s+1:v}$.
This representation enables rigorous probabilistic analysis of system survival, failure rate, and other performance metrics under various lifetime distributions. The $s$-out-of-$v$ structure has therefore become a fundamental model in reliability engineering due to its ability to incorporate redundancy while maintaining analytical tractability. Its importance continues to grow, especially in modern industrial systems where designing for partial functionality is essential for ensuring operational stability and fault tolerance. Two important variants of the $s$-out-of-$v$ structure commonly discussed in the literature are the $s$-out-of-$v$: M system and the $s$-out-of-$v$: G system. For illustration, a communication setup involving six transmitters may be examined.  If reliable message delivery requires that at least two transmitters remain active at any time, the configuration is represented as a $2$-out-of-$6{:}G$ system. More broadly, a system composed of $v$ independent and identically distributed strength components operating under a shared random stress is considered. The system is regarded as functional only when no fewer than $s$ of the $v$ component strengths exceed the imposed stress, where $1 \leq s \leq v$, and is thus termed an $s$-out-of-$v{:}G$ system. Let $X_1, \ldots, X_v$ denote the component strength variables governed by $F(\cdot)$, and let $Y$ represent the common stress variable with distribution function $G(\cdot)$. The resulting system reliability is expressed as
\begin{equation*}
    \begin{aligned}
    \Psi_{s,v} = & P(\text{no fewer than $s$ of } (X_1, X_2, \ldots, X_v)
 \text{ exceed } Y) \\
= & \sum_{i=s}^{v} \binom{v}{i} \int_{-\infty}^{\infty} [1 - F(t)]^{\,i} [F(t)]^{\,v-i} \, dG(t).
\end{aligned}
\end{equation*}
The concept of reliability in a multicomponent stress–strength model is first introduced by \citet{Bhattacharyya}, who laid the foundational framework for analyzing systems in which several components are subjected to a random stress. The multicomponent stress–strength reliability (MSSR) problem has attracted considerable attention, with numerous authors proposing classical estimation procedures, distribution-specific formulations, and theoretical advancements to characterize system reliability under various modeling scenarios. Over the years, researchers have investigated MSSR under a wide range of lifetime distributions, different censoring schemes (CSs), and inferential settings, demonstrating the model's versatility and practical relevance in engineering reliability analysis. Significant developments and related studies on MSSR can be found in the works of \citet{Kayal}, \citet{Singh_dp}, \citet{Kotb}, \citet{Nadar}, \citet{Singh_k}, \citet{Jana}, and \citet{Jha}, among others. These contributions encompass both classical and Bayesian frameworks, propose analytical expressions for reliability, explore PC and hybrid CSs, and introduce flexible lifetime distributions suited for real-world applications. Collectively, these studies have enriched the understanding of multicomponent stress–strength reliability and continue to motivate further research in system modeling, parametric inference, and reliability assessment.

Subsequently, censored data play a key role in reliability analysis when complete failure times are unavailable due to time, cost, or other measurement constraints. In this context, Type~I and Type~II censoring are widely used, where testing terminates at a fixed time point or after a fixed number of failures has been recorded, respectively. \citet{Oommen24} discusses various applications of these schemes. An important extension is the progressive Type~II (PT2)  censoring scheme (CS) initially proposed by \citet{herd}, which allows the pre-planned random withdrawal of surviving units during the experiment. Let a life test be conducted on $N$ systems, each comprising $K$ components. At the time the $i$th system fails, $S_i$ surviving systems are randomly withdrawn from the experiment. Similarly, at the time the $j$th component fails, $R_j$ surviving components are randomly removed from the experiment. This procedure continues sequentially until $n$th system failure recorded along with failure of $v$th component. Thus testing terminates at the observation of $n$th system failure and $v$th component failure. Accordingly the remaining
\[
S_n = N - n - \sum_{i=1}^{n-1} S_i \quad \text{and} \quad
R_v = V - v - \sum_{j=1}^{v-1} R_j
\]
systems and components are withdrawn from the experiment, and testing stops. Progressive censoring enhances experimental flexibility and efficiency while preserving statistical accuracy, as discussed in \citet{Balakrishnan book} and \citet{Balakrishnan censoring}. \citet{kohan19} discusses inference for MSSR under progressive censoring by considering the Kumaraswamy distribution.

Lifetime models with bounded support are important tools in statistical data analysis, especially when the variable of interest is restricted to the bounded interval $(0,1)$. Distributions defined on this interval are commonly known as unit lifetime distributions and are widely used to model proportions, rates, and reliability-related quantities. Alongside established models such as the Beta and Kumaraswamy distributions, several unit distributions have been developed by applying suitable transformations to existing lifetime models, thereby enhancing modeling flexibility. A new model, the unit-generalized Rayleigh (UGR) distribution, was recently introduced by \citet{Jha_unit}. The UGR distribution is obtained by applying a variable transformation to the generalized Rayleigh distribution. The new model provides additional flexibility for modeling bounded-lifetime data. We also refer to \citet{Kundu_csda} for a discussion. Its tractable form and adaptable shape make it a useful alternative for analyzing data supported on a bounded interval with potential applications in reliability analysis and related fields. Let $W$ be a random variable that follows a generalized Rayleigh distribution, then $T=e^{-W}$ follows UGR distribution with distribution function,
\begin{equation}\label{cdf}
    F_T(t) = 1- \big(1-e^{-(\varpi\ln t)^2}\big)^{\eta}; \quad0<t<1,\eta>0,\varpi>0,
\end{equation}
with $\eta$ and $\varpi$ are parameters. The corresponding density, survival and hazard rate functions are given by,
\begin{equation}
    f_T(t)=2\eta\varpi^2\left(-\frac{\ln t}{t}\right)e^{-(\varpi\ln t)^2}\left(1-e^{-(\varpi\ln t)^2}\right)^{\eta-1}; \quad0<t<1,\eta>0,\varpi>0,
\end{equation}
\begin{equation}
    S_T(t)= \big(1-e^{-(\varpi\ln t)^2}\big)^{\eta}; \quad0<t<1,\eta>0,\varpi>0,
\end{equation}
and 
\begin{equation}
    h_T(t)= \frac{2\eta\varpi^2\left(-\frac{\ln t}{t}\right)e^{-(\varpi\ln t)^2}}{\left(1-e^{-(\varpi\ln t)^2}\right)}; \quad0<t<1,\eta>0,\varpi>0.
\end{equation}
A notable feature of the UGR distribution is that different parameter assignment enables it to assume various shapes. We plot the density and distribution functions for various parameter values. We find useful patterns as depicted in Figure \ref{pdf_hrf}. It indicates increasing shapes, positive skew, negative skew, unimodal curves, and even bathtub-type behaviors. Due to this flexibility, the UGR model can fit many types of bounded data. Overall, it provides a useful and adaptable option for analyzing lifetime data on the unit interval, thus making it a valuable addition to reliability and statistical modeling.

\begin{figure}
    \centering
    \includegraphics[width=1\linewidth]{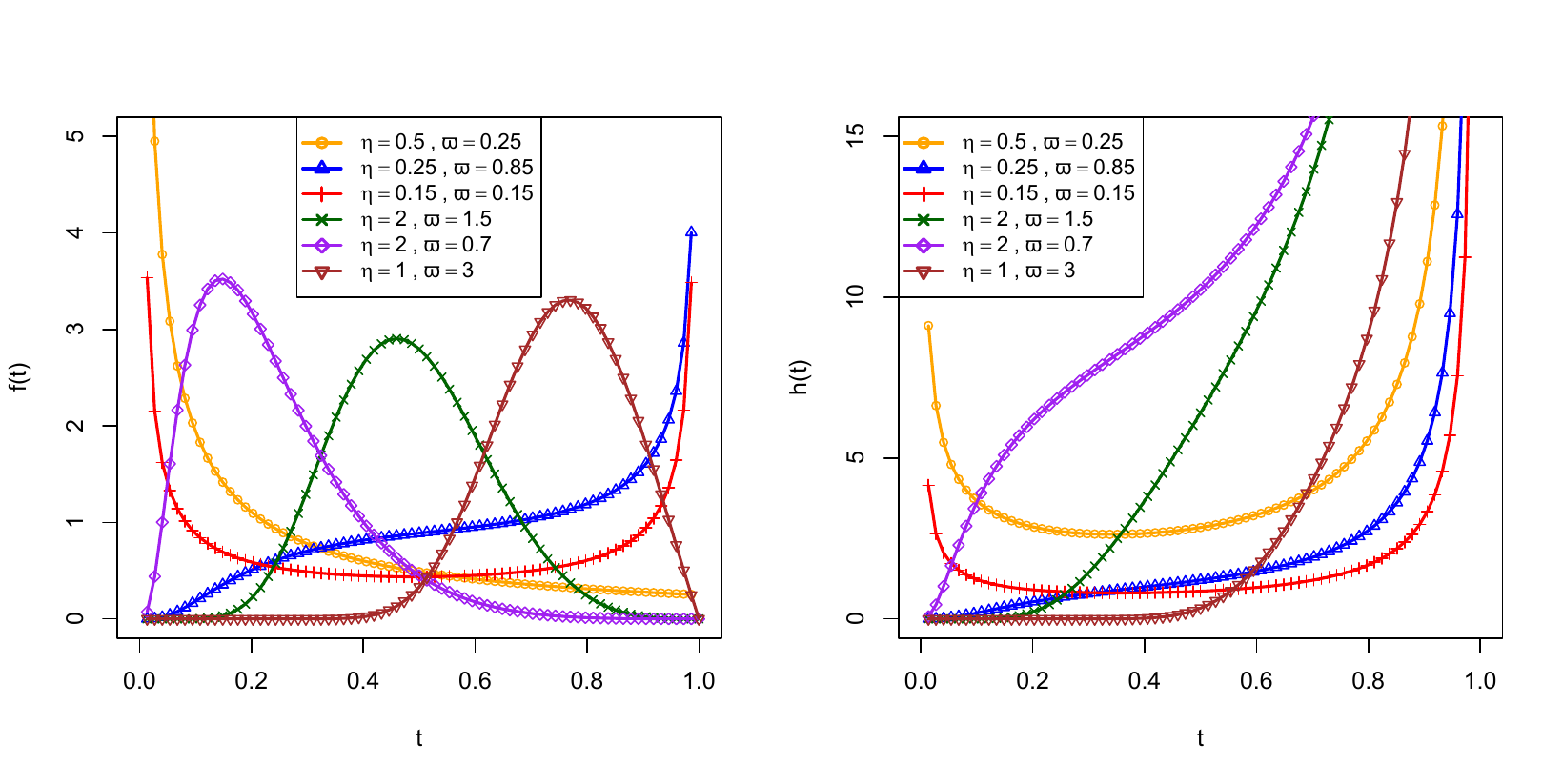}
    \caption{PDF and hazard plot for different values of $\eta$ and $\varpi$.}
    \label{pdf_hrf}
\end{figure}

We develop a unified statistical model for analyzing the stress-strength reliability in a multicomponent system under a PT2 censoring framework. We first derive the maximum likelihood estimates (MLEs) via the expectation-maximization (EM) algorithm in subsection \ref{mle-em}. Using the missing value principle, we derive Fisher information and asymptotic confidence interval (ACI) in Section \ref{FIM}. We study the maximum product of the spacing function as an alternative estimate to MLE in subsection \ref{mle-mps}. In Section \ref{bayesian}, Bayesian inference is discussed under informative and non-informative priors. In Section \ref{simulation study}, a comprehensive simulation study is presented to analyze the behavior of the proposed estimates. Real-life reliability data is discussed in Section \ref{application}. We propose optimal censoring plans based on three different comparative criteria in Section \ref{optimal cs}. Finally, we provide some concluding remarks and discuss the future direction of the recent study in Section \ref{conclusion}.











\section{Estimation of \texorpdfstring{$\Psi_{s,v}$}{Rs,v}}
\subsection{The EM algorithm}
The EM algorithm is a powerful procedure for computing MLEs of unknown parameters under censored samples, see \citet{Dempster}. It is widely used in likelihood estimation involving missing or incomplete data. This algorithm iteratively alternates between expectation (E) and maximization (M) steps. The first expected value of the complete-data log-likelihood is computed under the given current parameter estimates. Then maximization step updates current parametric estimates by maximizing the expected log-likelihood. For competing risks models with unknown causes of failure, the EM algorithm treats the unobserved causes as latent variables. In the E-step, the probability of each possible failure cause is computed under observations with unknown causes, based on current parameter estimates. The M-step updates estimates by maximizing the expected likelihood, effectively assigning the unknown causes in a statistically principled manner. This process continues until convergence, typically yielding maximum likelihood estimates of the model parameters. Several studies including \citet{Ng}, \citet{Swaroop}, \citet{Rastogi}, \citet{Shi}, and \citet{El‐Saeed} discuss several applications of this procedure.

Let $\boldsymbol{\Theta}$ be a vector of unknown parameters. Suppose the complete data is denoted by $\mathbf{x} = (x_1, x_2, \dots, x_n)$ with joint density $f(x_i; \boldsymbol{\Theta})$. The complete-data likelihood function is given by:
\[
L_c(\boldsymbol{\Theta} \mid \mathbf{x}) = \prod_{i=1}^{n} f(x_i; \boldsymbol{\Theta}).
\]

In practical scenarios, we often observe only part of the data. Let the observed portion be denoted by $\mathbf{y} = (y_1, \dots, y_m)$ and the missing portion by $\mathbf{z} = (z_{m+1}, \dots, z_n)$, such that $\mathbf{x} = (\mathbf{y}, \mathbf{z})$.

Let $\boldsymbol{\Theta}_m$ be the estimate of $\boldsymbol{\Theta}$ at the $m$-th iteration. The EM algorithm iteratively applies the following two steps:

\textbf{E-step:} Compute the expected value of the complete-data log-likelihood concerning the conditional distribution of the missing data given the observed data and the current parameter estimates:
\begin{align*}
    Q(\boldsymbol{\Theta} \mid \boldsymbol{\Theta}_m) 
= & E_{\mathbf{z} \mid \mathbf{y}, \boldsymbol{\Theta}_m} \left[ \log L_c(\boldsymbol{\Theta} \mid \mathbf{y}, \mathbf{z}) \right]\\
= & \int \log L_c(\boldsymbol{\Theta} \mid \mathbf{y}, \mathbf{z}) \, p(\mathbf{z} \mid \mathbf{y}, \boldsymbol{\Theta}_m) \, d\mathbf{z}.
\end{align*}

\textbf{M-step:} Maximize the expected log-likelihood computed in the E-step concerning the parameters $\boldsymbol{\Theta}$:
\[
\boldsymbol{\Theta}_{m+1} = \arg \max_{\boldsymbol{\Theta}} Q(\boldsymbol{\Theta} \mid \boldsymbol{\Theta}_m).
\]
This algorithm guarantees a non-decreasing sequence of likelihood values and typically converges to a local maximum of the likelihood function under the observed data. Its modular structure and numerical stability make it particularly well-suited for models involving latent or incomplete data like censored or masked competing risks scenarios.
\subsection{Maximum likelihood estimation (MLE)}
Here, we derive the MLE of the multicomponent stress-strength reliability $R_{s,v}$ when the latent failure times are observed under PT2 censoring. Let $X$ and $Y$ be strength and stress variables following $UGR(\eta_1,\varpi)$ and $UGR(\eta_2,\varpi)$ distributions, respectively, with unknown shape parameters $\eta_1,\eta_2$ and a common parameter $\varpi$. The reliability in a multicomponent stress-strength model under a unit Generalized Rayleigh distribution is given by
\begin{equation}
    \Psi_{s,v} = \frac{\eta_2}{\eta_1}\sum_{i=s}^v\binom{v}{i}B\left(\frac{\eta_2}{\eta_1}+i,v+1-i\right),
\end{equation}
where $B(.,.)$ stands for the Beta function. Figure \ref{mc_ssr} illustrates stress–strength reliability of a $v=5$ multicomponent system with varying $s$.

\begin{figure}[H]
    \centering
    \includegraphics[width=0.48\linewidth]{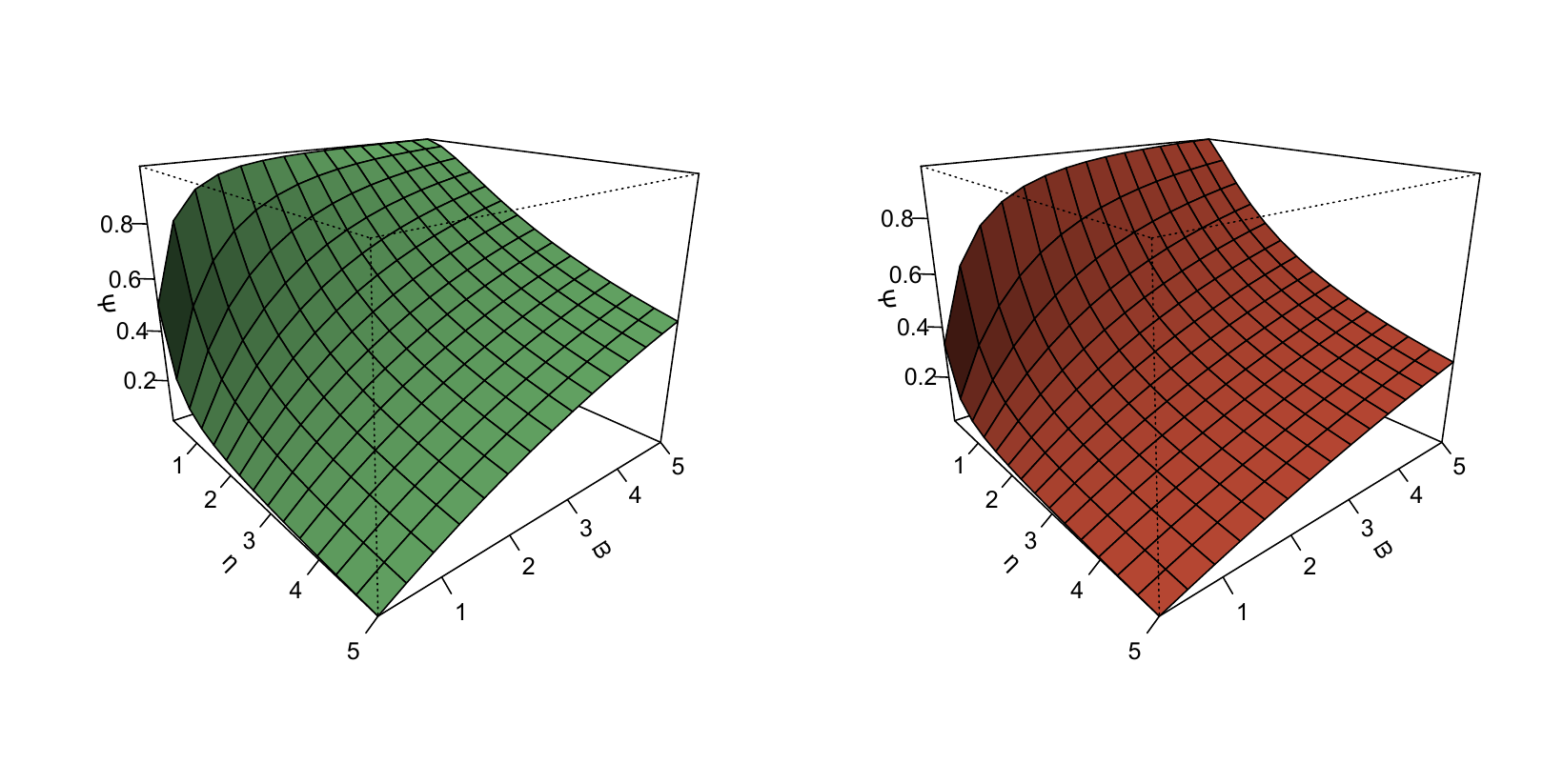}
   \caption{Stress-strength reliability for fixed $v=5$ with $s=3$ (left) and $s=4$ (right).}
    \label{mc_ssr}
\end{figure}
We first obtain MLEs of the parameters $\eta_1,\eta_2$ and $\varpi$ using PT2 censored samples. In constructing the likelihood function, $n$ systems are subjected to a life test, from which the following observations are obtained:

\begin{equation}\label{stress-strength-data}\underbrace{
    \begin{pmatrix}
        X_{11} & X_{12} & \dots & X_{1v}\\
        X_{21} & X_{22} & \dots & X_{2v}\\
        \vdots & \vdots & \ddots & \vdots\\
        X_{n1} & X_{n2} & \dots & X_{nv}
    \end{pmatrix} }_{\text{Observed strength variables}}
    \quad \text{and} \underbrace{
    \begin{pmatrix}
        Y_1\\
        Y_2\\
        \vdots\\
        Y_n
    \end{pmatrix}}_{\text{Observed stress variables}}
\end{equation}
Now, consider $\{X_{i1}, X_{i2},\dots, X_{iv}\}$, $i=1,2,\dots,n$, be the observed PT2 censored samples from $UGR(\eta_2,\varpi)$ and $\{Y_1,Y_2,\dots,Y_n\}$ be the PT2 censored samples from $UGR(\eta_2,\varpi)$ with the CSS $(V, v, R_1, R_2, \dots, R_v)$ and  $(N, n, S_1, S_2, \dots, S_n)$ respectively. Now the likelihood function becomes,
\begin{equation}
    \mathcal{L}(\eta_1,\eta_2,\varpi;x,y)=C_1\prod_{i=1}^n\left( C_2 \prod_{j=1}^v f(x_{ij})\left[1-F(x_{ij})\right]^{R_j}\right)f(y_i)\left[1-F(y_i)\right]^{S_i}
\end{equation}
where constants $C_1$ and $C_2$ are given by,
$$C_1 = N(N − S_1 − 1)\dots(N − S_1 − \dots − S_{n−1} − n + 1),$$ and $$C_2 = V(V − R_1 − 1)\dots(V − R_1 − \dots − R_{v−1} − v + 1).$$

The corresponding log-likelihood function is given by
\begin{equation}\label{observed log-lik}
    \begin{aligned}
        \ell_1(\eta_1,\eta_2,\varpi) = & nv\ln \eta_1+n\ln{\eta_2} + 2n(v+1)\ln{\varpi} -\varpi^2\left\{\sum_{i=1}^n
        \sum_{j=1}^v\left(\ln x_{ij}\right)^2 + \sum_{i=1}^n\left(\ln y_i\right)^2\right\} \\
        & + \sum_{i=1}^n\sum_{j=1}^v\left(\eta_1(1+R_j)-1 \right)\ln\left(1-e^{-(\varpi\ln x_{ij})^2}\right) \\
        & + \sum_{i=1}^n\left(\eta_2(1+S_i)-1 \right)\ln\left(1-e^{-(\varpi\ln y_i)^2}\right)\\
        & + \sum_{i=1}^n\sum_{j=1}^v \ln\left(\frac{\ln x_{ij}}{x_{ij}}\right) + \sum_{i=1}^n\ln\left(\frac{\ln y_i}{y_i}\right).
    \end{aligned}
\end{equation}
Thus, likelihood equations are obtained as follows:
\begin{equation}\label{ml-alp1}
    \frac{\partial\ell_1}{\partial\eta_1} = \frac{nv}{\eta_1} + \sum_{i=1}^n\sum_{j=1}^v(1+R_j)\ln\left(1-e^{-(\varpi\ln x_{ij})^2}\right),
\end{equation}
\begin{equation}\label{ml-alp2}
    \frac{\partial\ell_1}{\partial\eta_2} = \frac{n}{\eta_2} + \sum_{i=1}^n(1+S_i)\ln\left(1-e^{-(\varpi\ln y_i)^2}\right),
\end{equation}
\begin{equation}\label{ml-bet}
    \begin{aligned}
        \frac{\partial\ell_1}{\partial\varpi} = & \frac{2n(v+1)}{\varpi} -2\varpi\left\{\sum_{i=1}^n
        \sum_{j=1}^v\left(\ln x_{ij}\right)^2 + \sum_{i=1}^n\left(\ln y_i\right)^2 \right. \\
        & \left. - \sum_{i=1}^n\sum_{j=1}^v\left(\eta_1(1+R_j)-1 \right) \frac{e^{-(\varpi\ln x_{ij})^2}\left(\ln x_{ij}\right)^2}{1-e^{-(\varpi\ln x_{ij})^2}} 
        - \sum_{i=1}^n\left(\eta_2(1+S_i)-1 \right) \frac{e^{-(\varpi\ln y_i)^2}\left(\ln y_i\right)^2}{1-e^{-(\varpi\ln y_i)^2}}\right\}.
    \end{aligned}
\end{equation}
The MLEs $\hat\eta_1$, $\hat\eta_2$, and $\hat\varpi$ of $\eta_1$, $\eta_2$, and $\varpi$ are obtained by solving the three non-linear equations (\ref{ml-alp1}), (\ref{ml-alp2}), and (\ref{ml-bet}). However, these equations do not admit closed-form solutions; numerical methods are necessary to compute the estimates. In particular, the Newton–Raphson algorithm is employed for this purpose. However, this method is often susceptible to the initial values chosen, which can lead to convergence issues. To address the challenges associated with direct maximization of the likelihood function (\ref{observed log-lik}), we turn to the EM algorithm, which is discussed in the next section.

\subsubsection{MLE via the EM}\label{mle-em}
Suppose $X=(X_{i1},X_{i2},\dots,X_{iv})$ denote observed and $Z=(Z_{i1},Z_{i2},\dots,Z_{iv})$ denote censored strength data. Note that $Z_{ij}$ is $1\times R_j$ vector denoting $(Z_{ij1},Z_{ij2},\dots,Z_{ijR_j})$ for $i=1,2,\dots,n$ and $j=1,2,\dots,v$. Similarly let $Y=(Y_1,Y_2,\dots,Y_n)$ be observed and $W=(W_{1},W_{2},\dots,W_n)$ be the censored stress data. Also note that, $W_i$ is $1\times S_i$ vector denoting $(W_{i1},W_{i2},\dots,W_{iS_i})$ for $i=1,2,\dots,n$. Observe that $(X, Z)$ and $(Y, W)$ constitute the complete data for the strength and stress variables. Define the complete data set as $\mathcal{D}_c = \{(X, Z),\, (Y, W)\}$. Accordingly, upon omitting the additive constant, the log-likelihood function corresponding to $\mathcal{D}_c$ is given by

\begin{equation}
     \begin{aligned}
         \ell_1^c(\eta_1,\eta_2,\varpi) = & Vn\ln \eta_1 + N\ln \eta_2 + 2(Vn+N)\ln \varpi \\
         & - \varpi^2\left\{\sum_{i=1}^n\sum_{j=1}^v\left(\ln x_{ij}\right)^2 + \sum_{i=1}^n\sum_{j=1}^v\sum_{l=1}^{R_j}\left(\ln z_{ijl}\right)^2 + \sum_{i=1}^n\left(\ln y_i\right)^2 + \sum_{i=1}^n\sum_{u=1}^{S_i}\left(\ln w_{iu}\right)^2\right\} \\
        & + (\eta_1-1) \left\{\sum_{i=1}^n\sum_{j=1}^v\ln\left(1-e^{-(\varpi\ln x_{ij})^2}\right) + \sum_{i=1}^n\sum_{j=1}^v\sum_{l=1}^{R_j} \ln\left(1-e^{-(\varpi\ln z_{ijl})^2}\right)\right\} \\
        & + (\eta_2-1) \left\{\sum_{i=1}^n\ln\left(1-e^{-(\varpi\ln y_{i})^2}\right) + \sum_{i=1}^n\sum_{u=1}^{S_i} \ln\left(1-e^{-(\varpi\ln w_{iu})^2}\right)\right\}.
     \end{aligned}
\end{equation}

\begin{lemma}\label{lemma1}
    Consider a PT2 censored sample $X_1, X_2,\dots, X_n$ from $UGR(\eta,\varpi)$  distribution. Then, given $Z_{x_i}>X_i$ and $i=1,2,\dots,n$, we have the following results,
    \begin{enumerate}[(i)]
        \item $E\left[\left(\ln z_{x_i}\right)^2 \mid z_{x_i} > x_{i}\right] = -\frac{\eta}{\varpi^2(1-F(x_{i}))}\int_0^{1-e^{-(\varpi\ln x_{i})^2}} \ln(1-u)u^{\eta-1}du,\quad i=1,2,\dots,n.$
        \item $E\left[\ln\left(1-e^{-(\varpi\ln z_{x_i})^2}\right) \mid z_{x_i} > x_{i}\right] = \ln\left(1-e^{-(\varpi\ln x_{i})^2}\right)-\frac{1}{\eta},\quad i=1,2,\dots,n.$
    \end{enumerate}
\end{lemma}
\begin{proof}
    Proof is given in Appendix \hyperref[appendix-1]{1}.
\end{proof}
Now pseudo-log-likelihood function is computed. Ignoring the constant terms, it is denoted by $\ell_s(\eta_1,\eta_2,\varpi)$. Now observe that
\begin{equation}
     \begin{aligned}
         \ell_1^s(\eta_1,\eta_2,\varpi) = & Vn\ln \eta_1 + N\ln \eta_2 + 2(Vn+N)\ln \varpi - \varpi^2\left\{\sum_{i=1}^n\sum_{j=1}^v\left(\ln x_{ij}\right)^2 + \sum_{i=1}^n\left(\ln y_i\right)^2 \right\} \\
        & + (\eta_1-1) \left\{\sum_{i=1}^n\sum_{j=1}^v\ln\left(1-e^{-(\varpi\ln x_{ij})^2}\right)\right\} + (\eta_2-1) \left\{\sum_{i=1}^n\ln\left(1-e^{-(\varpi\ln y_{i})^2}\right) \right\} \\
        & - \varpi^2 \left\{ \sum_{i=1}^n\sum_{j=1}^v\sum_{l=1}^{R_j}E\left[\left(\ln z_{ijl}\right)^2 \mid z_{ijl} > x_{ij}\right] 
        + \sum_{i=1}^n\sum_{u=1}^{S_i}E\left[\left(\ln w_{iu}\right)^2 \mid w_{iu} > y_i\right]\right\} \\
        & + (\eta_1-1) \left\{\sum_{i=1}^n\sum_{j=1}^v\sum_{l=1}^{R_j} E\left[\ln\left(1-e^{-(\varpi\ln z_{ijl})^2}\right) \mid z_{ijl} > x_{ij}\right]\right\}\\
        & +  (\eta_2-1) \left\{\sum_{i=1}^n\sum_{u=1}^{S_i} E\left[\ln\left(1-e^{-(\varpi\ln w_{iu})^2}\right) \mid w_{iu} > y_i \right]\right\}.
     \end{aligned}
\end{equation}
The expectations appearing in the above equation are evaluated using Lemma \ref{lemma1}, and explicit expressions are provided in Appendix \hyperref[appendix-2]{2}.

Additionally, during the M-step of the EM algorithm, 
$\ell_s(\eta_1,\eta_2,\varpi)$ is maximized with respect to 
$\eta_1,\eta_2,$ and $\varpi$. Let parameter estimates at the 
$m$-th iteration be denoted by 
$\left(\eta_1^{(m)},\eta_2^{(m)},\varpi^{(m)}\right)$. 
Then, the updated estimates 
$\left(\eta_1^{(m+1)},\eta_2^{(m+1)},\varpi^{(m+1)}\right)$ 
are obtained by maximizing the following expression:

\begin{equation}
     \begin{aligned}
         \ell_1^{obj}(\eta_1,\eta_2,\varpi) = & Vn\ln \eta_1 + N\ln \eta_2 + 2(Vn+N)\ln \varpi - \varpi^2\left\{\sum_{i=1}^n\sum_{j=1}^v\left(\ln x_{ij}\right)^2 + \sum_{i=1}^n\left(\ln y_i\right)^2 \right\} \\
        & + (\eta_1-1) \left\{\sum_{i=1}^n\sum_{j=1}^v\ln\left(1-e^{-(\varpi\ln x_{ij})^2}\right)\right\} + (\eta_2-1) \left\{\sum_{i=1}^n\ln\left(1-e^{-(\varpi\ln y_{i})^2}\right) \right\} \\
        & - \varpi^2 \left\{ \sum_{i=1}^n\sum_{j=1}^v R_j A_1(x_{ij};\eta_1^m,\varpi^m) + \sum_{i=1}^n S_i A_2(y_i;\eta_2^m,\varpi^m)\right\} \\
        & + (\eta_1-1) \left\{\sum_{i=1}^n\sum_{j=1}^vR_j B_1(x_{ij};\eta_1^m,\varpi^m)\right\}  +  (\eta_2-1) \left\{\sum_{i=1}^mS_i B_2(y_i;\eta_2^m,\varpi^m)\right\}.
     \end{aligned}
\end{equation}

Therefore, we have the updated estimate $\varpi^{(m+1)}$ of $\varpi$, which can be computed by solving the equation
\begin{equation}
    \mathcal{H}(\varpi)=0,
\end{equation}
where 
\begin{equation*}
    \begin{aligned}
        \mathcal{H}(\varpi)= & \frac{2(Vn+N)}{\varpi}-2\varpi\left\{\sum_{i=1}^n\sum_{j=1}^v\left(\ln x_{ij}\right)^2 + \sum_{i=1}^n\left(\ln y_i\right)^2 - \sum_{i=1}^n\sum_{j=1}^v\left(\hat\eta_1(\varpi)-1 \right) \frac{e^{-(\varpi\ln x_{ij})^2}\left(\ln x_{ij}\right)^2}{1-e^{-(\varpi\ln x_{ij})^2}} \right. \\
       & \left. - \sum_{i=1}^n\left(\hat\eta_2(\varpi)-1 \right) \frac{e^{-(\varpi\ln y_i)^2}\left(\ln y_i\right)^2}{1-e^{-(\varpi\ln y_i)^2}} + \sum_{i=1}^n\sum_{j=1}^v R_jA_1(x_{ij};\eta_1^m,\varpi^m) + \sum_{i=1}^mS_iA_2(y_i;\eta_2^m,\varpi^m)  \right\}. 
    \end{aligned}
\end{equation*}
with 
$$\hat\eta_1(\varpi)=-\frac{Vn}{\sum_{i=1}^n\sum_{j=1}^v\left[\ln\left(1-e^{-(\varpi\ln x_{ij})^2}\right) +R_jB_1(x_{ij};\eta_1^m,\varpi^m)\right]},$$
$$\hat\eta_2(\varpi)=-\frac{N}{\sum_{i=1}^n\left[\ln\left(1-e^{-(\varpi\ln y_i)^2}\right) + S_iB_2(y_i;\eta_2^m,\varpi^m)\right]}.$$
Subsequently, the updated estimates of $\eta_1$ and $\eta_2$ can be computed using the expressions $\eta_1^{m+1}=\hat\eta_1(\varpi^{m+1})$ and $\eta_2^{m+1}=\hat\eta_2(\varpi^{m+1})$ respectively. Therefore MLE of $\Psi_{s,v}$ is obtained as 
\begin{equation}
    \hat\Psi_{s,v}^{EM} = \frac{\hat\eta_2}{\hat\eta_1}\sum_{i=s}^v\binom{v}{i}B\left(\frac{\hat\eta_2}{\hat\eta_1}+i,v+1-i\right).
\end{equation}

\subsection{Fisher information matrix (FIM)}\label{FIM}
In this section, we compute observed and expected Fisher information matrices using the missing information principle as discussed in \citet{louis}. The observed FIM plays an inferential role for obtaining ACI of unknown parameters $\Theta=(\eta_1,\eta_2,\varpi)$ and also any parametric function of $\Theta$, whereas the expected FIM is used for constructing optimal censoring plans.

\subsubsection{Observed Fisher information matrix}\label{OFIM}
Let $\mathcal{I}_{X,Y}(\Theta)$ denote observed information, $\mathcal{I}_{\mathcal{D}_c}(\Theta)$ be the complete information, and $\mathcal{I}_{\mathcal{D}_c\mid X,Y}(\Theta)$ denotes the missing information with $(\mathcal{D}_c\mid X,Y)$ as missing data. Thus, we have:
\begin{equation}\label{observed_mat}
    \mathcal{I}_{X,Y}(\Theta) = \mathcal{I}_{\mathcal{D}_c}(\Theta)-\mathcal{I}_{\mathcal{D}_c\mid X,Y}(\Theta).
\end{equation} 
The complete information matrix is
$$\mathcal{I}_{\mathcal{D}_c}(\Theta)=-E\left[\frac{\partial^2\ell_c(\mathcal{D}_c;\Theta)}{\partial\Theta^2}\right].$$
Furthermore, the Fisher information matrix corresponding to the unobserved observations is formulated as
$$\mathcal{I}_{\mathcal{D}_c \mid X,Y}(\Theta)= -E\left[\frac{\partial^2\ell_c(\mathcal{D}_c\mid X,Y,\Theta)}{\partial\Theta^2}\right].$$
Both $\mathcal{I}_{\mathcal{D}_c}(\Theta)$ and $\mathcal{I}_{\mathcal{D}_c \mid X,Y}(\Theta)$ are $3\times3$ real square matrices, where individual elements are $a_{ij}(\eta_1,\eta_2,\varpi)$ and $b_{ij}(\eta_1,\eta_2,\varpi)$ for $i,j=1,2,3$ respectively. The corresponding explicit expressions are given in Appendix \hyperref[appendix-2]{2}.

\begin{theorem}\label{theorem1}
    Suppose, $\hat \Psi_{s,v}^{EM}$ denotes MLE of $\Psi_{s,v}$, then we have 
    $$\left(\hat \Psi_{s,v}^{EM} - \Psi_{s,v}\right)\xrightarrow{D}N\left(0,\mathcal{V}_{\hat \Psi_{s,v}^{EM}}\right),$$
    where $$\mathcal{V}_{\hat \Psi_{s,v}^{EM}}=\left(\frac{\partial\Psi_{s,v}}{\partial\eta_1}\right)^2\mathcal{I}_{11}^{-1}(\hat\Theta^{EM}) + \left(\frac{\partial\Psi_{s,v}}{\partial\eta_2}\right)^2\mathcal{I}_{22}^{-1}(\hat\Theta^{EM}) + 2\left(\frac{\partial\Psi_{s,v}}{\partial\eta_1}\right)\left(\frac{\partial\Psi_{s,v}}{\partial\eta_2}\right)\mathcal{I}_{12}^{-1}(\hat\Theta^{EM})$$ with $\hat{\Theta}^{EM}$ denoting MLE of $\Theta$, and $\mathcal{I}^{-1}(\hat\Theta^{EM})$ being inverse of observed Fisher information matrix evaluated at $\hat{\Theta^{EM}}$.
\end{theorem}
\begin{proof}
    Proof is given in Appendix \hyperref[appendix-1]{1}.
\end{proof}

Therefore, by applying Theorem \ref{theorem1}, the asymptotic $100(1 − \gamma)\%$ two-sided confidence interval for $\Psi_{s,v}$ is given by
\[\left(\hat \Psi_{s,v}^{EM}\pm q_{\gamma/2}\sqrt{\mathcal{V}_{\hat \Psi_{s,v}^{EM}}}\right),\]
where, $q_{\gamma/2}$ denotes the upper $(\gamma/2)$th quantile of the standard normal distribution $N(0, 1)$.

\subsubsection{Expected information matrix}\label{EFIM}
This subsection outlines the expected Fisher information structure arising from PT2 censored data. The second partial derivatives of the log-likelihood function (\ref{observed log-lik}) with respect to
$(\eta_1, \eta_2, \varpi)$ are obtained as
$$\frac{\partial^2\ell_1}{\partial\eta_1^2}=-\frac{nv}{\eta_1}, \quad\frac{\partial^2\ell_1}{\partial\eta_2^2}=-\frac{n}{\eta_2}, \quad \frac{\partial^2\ell_1}{\partial\eta_1\partial\eta_2}=\frac{\partial^2\ell_1}{\partial\eta_2\partial\eta_1}=0,$$
$$\frac{\partial^2\ell_1}{\partial\eta_1\partial\varpi}=\sum_{i=1}^n\sum_{j=1}^v2\varpi(1+R_j)\frac{e^{-(\varpi\ln x_{ij})^2}\left(\ln x_{ij}\right)^2}{1-e^{-(\varpi\ln x_{ij})^2}}, \quad
\frac{\partial^2\ell_1}{\partial\eta_2\partial\varpi}=\sum_{i=1}^n2\varpi(1+S_i) \frac{e^{-(\varpi\ln y_i)^2}\left(\ln y_i\right)^2}{1-e^{-(\varpi\ln y_i)^2}},$$
\begin{align*}
    \frac{\partial^2\ell_1}{\partial\varpi^2}= & -\frac{2n(v+1)}{\varpi^2} -2\left\{\sum_{i=1}^n
        \sum_{j=1}^v\left(\ln x_{ij}\right)^2 + \sum_{i=1}^n\left(\ln y_i\right)^2 \right\} \\
        & + 2\sum_{i=1}^n\sum_{j=1}^v\left(\eta_1(1+R_j)-1 \right) \frac{e^{-(\varpi\ln x_{ij})^2}\left(\ln x_{ij}\right)^2\left(1-e^{-(\varpi\ln x_{ij})^2}-2\varpi^2\left(\ln x_{ij}\right)^2\right)}{\left(1-e^{-(\varpi\ln x_{ij})^2}\right)^2} \\
        & + 2\sum_{i=1}^n\left(\eta_2(1+S_i)-1 \right) \frac{e^{-(\varpi\ln y_{i})^2}\left(\ln y_{i}\right)^2\left(1-e^{-(\varpi\ln y_{i})^2}-2\varpi^2\left(\ln y_{i}\right)^2\right)}{\left(1-e^{-(\varpi\ln y_i)^2}\right)^2}.
\end{align*}

\begin{lemma}\label{lemma2}
    Consider a PT2 censored sample $X_1, X_2,\dots, X_n$ from $UGR(\eta,\varpi)$ distribution  with CS $(N,n,R_1,R_2,\dots,R_n)$. Then we have the following results.
    \begin{enumerate}[(i)]
        \item $E\left[ \left(\ln x_{i}\right)^2 \right]=-\frac{\eta C_{i-1}}{\varpi^2}\sum_{l=1}^ia_{i,l}B(\eta\gamma_l,1)[\psi(1)-\psi(1+\eta\gamma_l)],$
        \item $E\left[\frac{e^{-(\varpi\ln x_i)^2}\left(\ln x_i\right)^2}{1-e^{-(\varpi\ln x_i)^2}}\right]=-\frac{\eta C_{i-1}}{\varpi^2}\sum_{l=1}^ia_{i,l}B(\eta\gamma_l-1,2)[\psi(2)-\psi(1+\eta\gamma_l)],$
        \item $E\left[\frac{e^{-(\varpi\ln x_{i}i)^2}\left(\ln x_{i}\right)^2\left(1-e^{-(\varpi\ln x_{i}i)^2}-2\varpi^2\left(\ln x_{i}\right)^2\right)}{\left(1-e^{-(\varpi\ln x_i)^2}\right)^2}\right] = -\frac{\eta C_{i-1}}{\varpi^2}\sum_{l=1}^ia_{i,l}\bigg[B(\eta\gamma_l-1,2)[\psi(2)-\psi(\eta\gamma_l+1)] + B(\eta\gamma_l-2,2)[(\psi(2)-\psi(\eta\gamma_l))^2 + \psi^{'}(2)-\psi^{'}(\eta\gamma_l)]\bigg],$
    \end{enumerate}
    where $\gamma_l=v-l+1+\sum_{l=1}^vR_l, C_{i-1}=\prod_{l=1}^i\gamma_l,a_{i,l}=\prod_{q=1,q\neq l}^i\frac{1}{\gamma_q-\gamma_l},$ for $i=1,2,\dots,n$. Also $B(x,y)$ denotes the standard beta function, $\psi(x)$ and $\psi^{'}(x)$ are the digamma and trigamma functions, respectively.
\end{lemma}
\begin{proof}
    Proof is given in Appendix \hyperref[appendix-1]{1}.
\end{proof}

We now have the expected Fisher information matrix given by
\begin{equation}
    \mathcal{I}_E(\Theta)=
    \begin{pmatrix}
        e_{11} & e_{12} & e_{13} \\[5pt]
        e_{21} & e_{22} & e_{23} \\[5pt]
        e_{31} & e_{32} & e_{33} \\[5pt]
\end{pmatrix}
\end{equation}
where \[ e_{p,q} = -E\left[\frac{\partial^2  \ell}{\partial \Theta_p\partial\Theta_q}\right];\quad p,q=1,2,3.\] The expectations involved in the matrix are derived using the Lemma \ref{lemma2}, and its elements are given in Appendix \hyperref[appendix-2]{2}. 

\subsection{Maximum product of spacing (MPS) estimation}\label{mle-mps}
The MPS method is initially discussed in \citet{cheng} and \citet{Ranneby}. It is a useful alternative to likelihood estimation.\\
Suppose $X_{(1)}, \ldots, X_{(n)}$ represent the order statistics from a population characterized by a continuous distribution function $F(x;\boldsymbol{\theta})$, with $\boldsymbol{\theta} \in \Theta \subset \mathbb{R}^p$. By setting $F(x_{(0)};\boldsymbol{\theta}) = 0$ and $F(x_{(n+1)};\boldsymbol{\theta}) = 1$, the spacings are obtained as

\[
D_i(\boldsymbol{\theta}) = F(x_{(i)}; \boldsymbol{\theta}) - F(x_{(i-1)}; \boldsymbol{\theta}), \quad \text{for } i = 1, 2, \ldots, n+1,
\] where $\sum\limits_{i=1}^{n+1}D_i(\theta)=1$. \citet{dutta24} describes useful applications of this method.

The objective function promotes uniformity of the probability–integral transform and minimizes the discrepancy between empirical and theoretical probability spacings, providing an information-theoretic foundation for robust parametric inference. So, based on the PT2 censored stress-strength data, we have proposed the product of the spacing function as
\begin{equation}
    \begin{aligned}
        \mathcal{Q}(\eta_1,\eta_2,\varpi) = & \prod_{i=1}^n\left(F(x_{i1})[1-F(x_{iv})]\prod_{j=2}^v[F(x_{ij})-F(x_{ij-1})]\right)\prod_{i=1}^n\prod_{j=2}^v[1-F(x_{ij})]^{R_j}\\
        & \times F(y_{1})[1-F(y_{n})]\prod_{i=2}^n[F(y_{i})-F(y_{i-1})]\prod_{i=1}^n[1-F(y_{i})]^{S_i}.
    \end{aligned}
\end{equation}
The corresponding log-likelihood function,  $\ell_2(\eta_1,\eta_2,\varpi) = \ln \mathcal{Q}(\eta_1,\eta_2,\varpi)$, is then given by:
\begin{equation}
    \begin{aligned}
        \ell_2(\eta_1,\eta_2,\varpi) = & \sum_{i=1}^n \left(\ln F(x_{i1}) + \ln [1-F(x_{iv})] + \sum_{j=2}^v \ln  [F(x_{ij})-F(x_{ij-1})] \right) \\
        & + \sum_{i=1}^n\sum_{j=1}^vR_j\ln [1-F(x_{ij})] + \ln F(y_{1}) + \ln [1-F(y_{n})] + \sum_{i=2}^n\ln [F(y_{i})-F(y_{i-1})] \\
        & + \sum_{i=1}^nS_i\ln [1-F(y_{i})].
    \end{aligned}
\end{equation}
The required estimates $\tilde{\eta_1}, \tilde{\eta_2}$ and $\tilde{\varpi}$ of $\eta_1,\eta_2$ and $\varpi$ are obtained by solving
$$ \frac{\partial\ell_2}{\partial\eta_1}=0;\quad \frac{\partial\ell_2}{\partial\eta_2}=0; \quad\text{and} \quad\frac{\partial\ell_2}{\partial\varpi}=0.$$
The resulting nonlinear equations are solved through a Newton-Raphson scheme, using the 
MLEs as starting values to ensure rapid and stable convergence. Consequently, the MPS estimates of $\Psi_{s,v}$ is obtained as
\begin{equation}
    \tilde\Psi_{s,v}^{MPS} = \frac{\tilde\eta_2}{\tilde\eta_1}\sum_{i=s}^v\binom{v}{i}B\left(\frac{\tilde\eta_2}{\tilde\eta_1}+i,v+1-i\right).
\end{equation}

Under standard regularity conditions, the consistency and asymptotic efficiency of the MPS estimator have been established by \citet{Anatolyev}, yielding:
\[
\left(\tilde{\Theta}^{MPS} - \Theta\right) \xrightarrow{D} N\left(0, \mathcal{I}_M^{-1}(\tilde\Theta^{MPS})\right),
\]

where $\mathcal{I}_M(\tilde\Theta^{MPS})$ is the observed FIM of MPS estimators, evaluated at $\tilde\Theta^{MPS}$. Thus, using the asymptotic properties of MPS, the ACIs for $\Psi_{s,v}$ are also constructed using this method.

\subsection{Bayesian inference}\label{bayesian}
In this section, Bayes estimates of parameters $\eta_1,\eta_2$, and $\varpi$ are obtained under different loss functions. Although the Bayes estimator is defined through the posterior distribution, its estimate depends on the considered loss function. The squared error (SE) loss treats overestimation and underestimation equally. In contrast, the LINEX loss accounts for asymmetric consequences. The SE is one of the popular loss functions defined as:
\begin{equation}
    L_{1}(\delta, \tilde{\delta})= (\delta - \tilde{\delta})^2.
\end{equation}
The Bayes estimate of $\zeta(\delta)$ of $\delta$ under this loss function is given by
\begin{equation}
    \hat{\zeta}_{SB}(\delta) = E[\zeta(\delta \mid \text{data})] =  \int_{0}^{\infty} \zeta(\delta) \pi(\delta \mid \text{data}) ~ d\delta.
\end{equation}
This loss function imposes an equal penalty for both overestimation and underestimation. In some cases, asymmetric loss functions have better applications. We consider the LINEX loss function, which is quite widely used in Bayesian inference. It is asymmetric in nature and defined as:
 \begin{equation}
     L_2(\delta, \tilde{\delta}) = e^{h(\tilde{\delta} - \delta)} -h(\tilde{\delta}-\delta) -1
 \end{equation}
where $h\ne 0$ controls the degree of asymmetry. The required estimate of $\zeta(\delta)$ is now given as
\begin{equation}
    \hat{\zeta}_{LB}(\delta) = -\frac{1}{h}\ln\left[E(e^{-h\zeta(\delta)})| \text{data}\right] = -\frac{1}{h}\ln\left[\int_0^\infty e^{-h\zeta(\delta)}\pi(\delta | \text{data})d\delta\right].
\end{equation}

The LINEX loss function accounts for asymmetric penalties for over- and under-estimation. It is well known for its efficacy in estimating parameters within the Bayesian framework. Employing these loss functions, we obtain Bayesian estimators of model parameters and evaluate their performance through simulations. Our analysis characterizes the behavior of estimators across various prior specifications, sample sizes, and posterior distributions. The findings reveal trade-offs between bias, accuracy, and robustness in relation to prior assumptions.

\subsection{Prior information}
\subsubsection{Informative prior}
The choice of prior distributions plays a crucial role in Bayesian estimation. A flexible family of prior distributions is often considered to strike a balance between analytical tractability and practical flexibility. We assume that parameters $\eta_1, \eta_2$, and $\varpi$ have independent gamma prior distributions.
The corresponding PDF of $\eta_1, \eta_2$, and $\varpi$ are given by:
\begin{equation*}
    \pi_1(\eta_1) = \frac{b_1^{a_1}}{\Gamma (a_1)} \eta_1^{a_1-1}e^{-b_1\eta_1};\quad \eta_1>0,a_1>0,b_1>0,
\end{equation*}
\begin{equation*}
    \pi_2(\eta_2) = \frac{b_2^{a_2}}{\Gamma (a_2)} \eta_2^{a_2-1}e^{-b_2\eta_2};\quad \eta_2>0,a_2>0,b_2>0,
\end{equation*}
and
\begin{equation*}
    \pi_3(\varpi) = \frac{b_3^{a_3}}{\Gamma (a_3)} \varpi^{a_3-1}e^{-b_3\varpi};\quad \varpi>0,a_3>0,b_3>0
\end{equation*}
respectively. Then the joint prior density becomes
\begin{equation}
    \pi(\eta_1,\eta_2,\varpi) \propto \eta_1^{a_1-1}\eta_2^{a_2-1}\varpi^{a_3-1}e^{-b_1\eta_1-b_2\eta_2-b_3\varpi}.
\end{equation}

As a result, we get two distinct posterior densities. Using likelihood function, the joint posterior density (JPDF) of $\eta,c$, and $d$ becomes
\begin{equation}\label{mle-post-inf}
    \begin{aligned}
    \Pi_1(\eta_1,\eta_2,\varpi\mid data)  \propto & \pi(\eta_1,\eta_2,\varpi\mid data)\times \text{Likelihood} \\
    = &  K_1^{-1} \eta_1^{nv+a_1-1}\eta_2^{n+a_2-1}\varpi^{2m(v+1)+a_3-1}\\
    & \times \exp\left[-b_1\eta_1-b_2\eta_2-b_3\varpi -\varpi^2\left\{\sum_{i=1}^n\sum_{j=1}^v\left(\ln x_{ij}\right)^2 + \sum_{i=1}^n\left(\ln y_i\right)^2\right\}\right]\\
    & \times \prod_{i=1}^n\prod_{j=1}^v\left(1-e^{-(\varpi\ln x_{ij})^2}\right)^{\left(\eta_1(1+R_j)-1 \right)} \prod_{i=1}^n\left(1-e^{-(\varpi\ln y_i)^2}\right)^{\left(\eta_2(1+S_i)-1 \right)},
\end{aligned}
\end{equation}
where
\begin{align*}
    K_1^{-1}= &\int_{0}^{\infty}\int_{0}^{\infty}\int_{0}^{\infty}  \eta_1^{nv+a_1-1}\eta_2^{n+a_2-1}\varpi^{2n(v+1)+a_3-1}\\
    & \times \exp\left[-b_1\eta_1-b_2\eta_2-b_3\varpi -\varpi^2\left\{\sum_{i=1}^n\sum_{j=1}^v\left(\ln x_{ij}\right)^2 + \sum_{i=1}^n\left(\ln y_i\right)^2\right\}\right]\\
    & \times \prod_{i=1}^n\prod_{j=1}^v\left(1-e^{-(\varpi\ln x_{ij})^2}\right)^{\left(\eta_1(1+R_j)-1 \right)} \prod_{i=1}^n\left(1-e^{-(\varpi\ln y_i)^2}\right)^{\left(\eta_2(1+S_i)-1 \right)} d\eta_1d\eta_2d\varpi.
\end{align*}
Further JPDF using spacing function is given by
\begin{equation}\label{mps-post-inf}
    \begin{aligned}
    \Pi_2(\eta_1,\eta_2,\varpi\mid data)  \propto & \pi(\eta_1,\eta_2,\varpi\mid data)\times \text{Spacing function} \\
    = &  K_2^{-1} \pi(\eta_1,\eta_2,\varpi\mid data)\prod_{i=1}^n\left(F(x_{i1})[1-F(x_{iv})]\prod_{j=2}^v[F(x_{ij})-F(x_{ij-1})]\right)\\
    & \times \prod_{i=1}^n\prod_{j=2}^v[1-F(x_{ij})]^{R_j} \times F(y_{1})[1-F(y_{n})]\prod_{i=2}^n[F(y_{i})-F(y_{i-1})]\\
    & \times \prod_{i=1}^n[1-F(y_{i})]^{S_i},
\end{aligned}
\end{equation}
where
\begin{align*}
    K_2^{-1}= \int_{0}^{\infty}\int_{0}^{\infty}\int_{0}^{\infty} & \pi(\eta_1,\eta_2,\varpi\mid data)\prod_{i=1}^n\left(F(x_{i1})[1-F(x_{iv})]\prod_{j=2}^v[F(x_{ij})-F(x_{ij-1})]\right)\\
    & \times \prod_{i=1}^n\prod_{j=2}^v[1-F(x_{ij})]^{R_j} \times F(y_{1})[1-F(y_{n})]\prod_{i=2}^n[F(y_{i})-F(y_{i-1})]\\
    & \times \prod_{i=1}^n[1-F(y_{i})]^{S_i} d\eta_1d\eta_2d\varpi.
\end{align*}
Due to the intractability of posterior distributions, closed-form Bayes estimators of $\eta_1,\eta_2$, and $\varpi$ under the considered loss functions are not obtained. Consequently, posterior samples are generated using a Markov chain Monte Carlo (MCMC) method. Specifically, the Metropolis–Hastings algorithm is employed to approximate posterior expectations. Accordingly, Bayes estimates under different loss functions are computed. MCMC methods are extensively used to generate samples from complex probability distributions when direct sampling is infeasible. In this study, posterior samples for the MLE and MPS methods are generated using the Metropolis–Hastings algorithm (see \citet{metropolis} and \citet{hastings}). The procedure outlined in Algorithm~\ref{Alg-1} is employed to obtain samples from the marginal posterior densities.

\subsubsection{Non informative prior}
In Bayesian analysis, we use a non-informative prior when proper information about parameters is available. The Jeffreys prior is quite useful in such a situation because it is built directly from the Fisher information and remains unchanged under reparameterization, see \citet{Jeffreys}. We use it for non-informative cases to ensure that our Bayesian estimates remain objective, driven mainly by the data, and free from unnecessary assumptions. This prior for unknown parameters $(\eta_1, \eta_2, \varpi)$ is derived based on the exact FIM $\mathcal{I}_{E}(\Theta)$. Thus we have
\begin{equation}
    \pi_J(\eta_1,\eta_2,\varpi)=\sqrt{\det(\mathcal{I}_E(\Theta))}.
\end{equation}
One may refer to \citet{santis01} in this regard. However, due to the complex form of the exact information matrix, the Jeffreys prior cannot be derived analytically. So, we consider a possible simplification to have a non-informative prior of the type (see \citet{Moala})
$$\pi_J(\eta_1,\eta_2,\varpi)=\pi(\varpi\mid\eta_1,\eta_2)\pi(\eta_1)\pi(\eta_2).$$ Using the Jeffreys rule, we have  $$\pi(\varpi\mid\eta_1,\eta_2) \propto\sqrt{-E\left(\frac{\partial^2 \ell_1}{\partial\varpi^2}\right)}.$$

Now, for $\pi(\eta_1)$ and $\pi(\eta_2)$, we choose a non-informative prior, for instance, a gamma prior with hyperparameters equal to $10^{-5}$. In this way, the non-informative Jeffreys prior for $(\eta_1,\eta_2,\varpi)$ is given by 
\begin{equation}\label{non-prior}
    \pi_J(\eta_1,\eta_2,\varpi)=\frac{1}{\varpi}\pi(\eta_1)\pi(\eta_2).
\end{equation}
Further the JPDF of $\eta_1,\eta_2$, and $\varpi$ turns out to be
\begin{equation}\label{mle-post-non}
    \begin{aligned}
    \Pi_{3}(\eta_1,\eta_2,\varpi\mid data)  \propto & \pi_J(\eta_1,\eta_2,\varpi)\times \text{Likelihood} \\
    = &  K_{3}^{-1} \pi_J(\eta_1,\eta_2,\varpi)\eta_1^{nv}\eta_2^{n}\varpi^{2n(v+1)}\\
    & \times \exp\left[-b_1\eta_1-b_2\eta_2-b_3\varpi -\varpi^2\left\{\sum_{i=1}^n\sum_{j=1}^v\left(\ln x_{ij}\right)^2 + \sum_{i=1}^n\left(\ln y_i\right)^2\right\}\right]\\
    & \times \prod_{i=1}^n\prod_{j=1}^v\left(1-e^{-(\varpi\ln x_{ij})^2}\right)^{\left(\eta_1(1+R_j)-1 \right)} \prod_{i=1}^n\left(1-e^{-(\varpi\ln y_i)^2}\right)^{\left(\eta_2(1+S_i)-1 \right)},
\end{aligned}
\end{equation}
where
\begin{align*}
    K_{3}^{-1}= &\int_{0}^{\infty}\int_{0}^{\infty}\int_{0}^{\infty}  \pi_J(\eta_1,\eta_2,\varpi)\eta_1^{nv}\eta_2^{n}\varpi^{2n(v+1)}\\
    & \times \exp\left[-b_1\eta_1-b_2\eta_2-b_3\varpi -\varpi^2\left\{\sum_{i=1}^n\sum_{j=1}^v\left(\ln x_{ij}\right)^2 + \sum_{i=1}^n\left(\ln y_i\right)^2\right\}\right]\\
    & \times \prod_{i=1}^n\prod_{j=1}^v\left(1-e^{-(\varpi\ln x_{ij})^2}\right)^{\left(\eta_1(1+R_j)-1 \right)} \prod_{i=1}^n\left(1-e^{-(\varpi\ln y_i)^2}\right)^{\left(\eta_2(1+S_i)-1 \right)} d\eta_1d\eta_2d\varpi.
\end{align*}
In a similar way, the JPDF using spacing function is derived as
\begin{equation}\label{mps-post-non}
    \begin{aligned}
    \Pi_{4}(\eta_1,\eta_2,\varpi\mid data)  \propto & \pi_J(\eta_1,\eta_2,\varpi\mid data)\times \text{Spacing function} \\
    = &  K_{4}^{-1} \pi_J(\eta_1,\eta_2,\varpi\mid data)\prod_{i=1}^n\left(F(x_{i1})[1-F(x_{iv})]\prod_{j=2}^v[F(x_{ij})-F(x_{ij-1})]\right)\\
    & \times \prod_{i=1}^n\prod_{j=2}^v[1-F(x_{ij})]^{R_j} \times F(y_{1})[1-F(y_{n})]\prod_{i=2}^n[F(y_{i})-F(y_{i-1})]\\
    & \times \prod_{i=1}^n[1-F(y_{i})]^{S_i},
\end{aligned}
\end{equation}
where
\begin{align*}
    K_{4}^{-1}= \int_{0}^{\infty}\int_{0}^{\infty}\int_{0}^{\infty} & \pi_J(\eta_1,\eta_2,\varpi\mid data)\prod_{i=1}^n\left(F(x_{i1})[1-F(x_{iv})]\prod_{j=2}^v[F(x_{ij})-F(x_{ij-1})]\right)\\
    & \times \prod_{i=1}^n\prod_{j=2}^v[1-F(x_{ij})]^{R_j} \times F(y_{1})[1-F(y_{n})]\prod_{i=2}^n[F(y_{i})-F(y_{i-1})]\\
    & \times \prod_{i=1}^n[1-F(y_{i})]^{S_i} d\eta_1d\eta_2d\varpi.
\end{align*}
Algorithm \ref{Alg-1} is applied to generate samples from the marginal posterior density \eqref{mle-post-non} and \eqref{mps-post-non} to obtain the required Bayes estimates. 

\begin{algorithm}[H]
\caption{M-H algorithm}
\label{Alg-1} 
\begin{description}
    \item [Step 1] Initialize $\Theta^0=(\eta_1^0, \eta_2^0,\varpi^0)$ and consider a multivariate normal distribution as the proposal density.
    \item [Step 2] Generate candidate parameters $\Theta^{*} = (\eta_1^*, \eta_2^*,\varpi^*)$ from the proposal density.
    \item [Step 3] Set $t = 1$.
    \item [Step 4] Compute the acceptance probability $h(\Theta^*, \Theta^{(t-1)})$ as: $$h(\Theta^*, \Theta^{(t-1)}) = \min\left(1, \frac{\Pi(\Theta^*\mid data)}{\Pi(\Theta^{(t-1)}\mid data)} \right) $$ where, \( \Pi(\Theta \mid data) \) is the joint posterior distribution.
    \item [Step 5] Generate a random number $u$ from uniform $U(0,1)$. With the probability $h$ (computed in Step 4), if $u \leq h$ accept \( \Theta^* \) as the new sample \( \Theta^{(t)} \) otherwise, keep \( \Theta^{(t-1)} \) as \( \Theta^{(t)} \).
    \item [Step 6] Compute $\Psi_{s,v}^{(t)}$ at $\Theta^{(t)}$.
    \item [Step 7] $t=t+1$.
    \item [Step 8] Repeat steps 3-7 $D$ times and obtain samples $\{\Psi_{s,v}^{(t)};~t =1, 2,3, \dots, D\}$.
    \item [Step 6] Discarding the initial $D_0$ samples, estimates of  $\Psi_{s,v}$ is obtained under different loss functions.  Accordingly, the Bayes estimates $\hat\Psi_{s,v}^{MH}$ of $\Psi_{s,v}$,  associated with the squared error and LINEX loss functions,  are respectively expressed as:

    \begin{equation*}
     \begin{aligned}
         \hat\Psi_{s,v}^{MH_1}  = & \frac{1}{D-D_0} \sum_{t=D_0+1}^{D} \Psi_{s,v}^{(t)} \\
        \hat\Psi_{s,v}^{MH_2}  = & -\frac{1}{h} \log \left(\frac{1}{D-D_0} \sum_{t=D_0+1}^{D} e^{-h \Psi_{s,v}^{(t)}}\right).
     \end{aligned}
    \end{equation*}
    The HPD credible interval for a  $\Psi_{s,v}$ is constructed by applying the procedure described in \citet{chen}.
  \end{description} 
\end{algorithm}

\section{Simulation study}\label{simulation study}
Extensive simulation studies are carried out using generated failure-time data to assess the performance of the proposed methods. The point and interval estimators are examined under various CSs. Estimation accuracy is evaluated using the average estimate (AE) and mean squared error (MSE), while interval estimators are compared in terms of their average length (AL) and coverage probability (CP).

Bayes estimators are computed under both informative and non-informative prior settings. For the informative case, the hyperparameters are chosen so that the prior means coincide with the true values of $\eta_1,\eta_2$, and $\varpi$. The results are obtained for different parameter combinations $(\eta_1,\eta_2,\varpi)=(1.9,2.1,3)$ and $(2.25,2.75,2)$. The hyperparameters are specified as $(b_1,b_2,b_3)=(1.5,2,3)$, and $(a_1,a_2,a_3)$. Posterior samples of $\eta_1,\eta_2$, and $\varpi$ are then generated using the Metropolis–Hastings algorithm based on the posterior distributions \eqref{mle-post-inf}, \eqref{mps-post-inf}, \eqref{mle-post-non}, and \eqref{mps-post-non}. Bayes estimates of $\Psi_{s,v}$ are subsequently obtained from $D=10{,}000$ iterations, with the first $D_0=2{,}000$ samples discarded as burn-in.
The proposed estimators are examined over various combinations of $N,n,V$, and $v$, with $s=2$, $s=3$, corresponding to the system reliability measures $\Psi_{2,v}$ and $\Psi_{3,v}$. All interval estimators are constructed at a nominal confidence level of $95\%$, and each simulation experiment is replicated $1{,}000$ times to ensure numerical stability and reliable inference. Summary measures for both point and interval estimation are reported in the corresponding tables.

Additionally, sample generation from the UGR distribution under a multicomponent system is carried out using the procedure described in Algorithm~\ref{Alg-2}. 
\begin{algorithm}[H]
\caption{Generation of PT2 censored data for UGR distribution}
\label{Alg-2} 
\begin{description}
    \item [Step 1] Start with the true value of parameters. We consider the following CSs along with two sets of censoring times:
\begin{align*}
    \text{Scheme 1:} \quad &  R_{j}=
	\begin{cases}
		V-v, & j =1\\
		0, & \text{otherwise}
	\end{cases} \quad \text{and} \quad
    S_{i}=
	\begin{cases}
		N-n, & i =1\\
		0, & \text{otherwise},
	\end{cases}
	\\
	\text{Scheme 2:} \quad &  R_{j}=
	\begin{cases}
		V-v, & j=v\\
		0, & \text{otherwise}
	\end{cases} \quad \text{and} \quad
    S_{i}=
	\begin{cases}
		N-n, & i =n\\
		0, & \text{otherwise}.
	\end{cases}
\end{align*}
    \item [Step 2] Using the algorithm proposed by \citet{Balakrishnan algo}, $v$ independent PT2 censored samples are generated from the Uniform $((0,1)$ distribution and denoted by $U_{1}, U_{2}, \ldots, U_{v}$.

    \item [Step 3] Obtain PT2 censored data from the strength variable in a single component system $x_{1}, x_{2}, \dots, x_{v}$ using the inverse transformation from the CDF (\ref{cdf}) by:
    \begin{equation*}
            x_{j} = \exp\left\{-\frac{1}{\varpi}\sqrt{-\ln{(1-(1-u_j)^{1/\eta_1})}}\right\}  ;\quad j=1,2, \dots,v.
    \end{equation*} 
     PT2 censored multicomponent data are constructed by embedding an $m$-out-of-$M$ system within a $v$-out-of-$v$ component structure.  Correspondingly, PT2 censored observations are obtained for the stress variable.

  \end{description}
\end{algorithm}
The numerical findings reported in Tables~\ref{sim point1}--\ref{sim length2} highlight several consistent features across different estimation methods. Better MSE performance of MLE, MPS, and Bayesian estimators is observed as the sample size increases, supporting large-sample validation. For fixed sample sizes, Bayesian estimators generally outperform MLE and MPS estimates in terms of MSEs, indicating higher estimation efficiency. Furthermore, larger samples are associated with shorter HPD and ACI intervals and improved coverage probabilities, demonstrating enhanced precision and robustness of statistical inference.

A comparison of MLE and MPS estimators shows that their performances are largely equivalent, with no method consistently dominating the other. In certain scenarios, MPS yields lower MSEs with shorter interval lengths and improved numerical stability. These findings suggest that MPS offers a competitive and reliable alternative to MLE, although it does not consistently outperform for all tabulated schemes.

Overall, the accuracy of both point and interval estimators improves with increasing sample sizes, with Bayesian methods performing most effectively in the majority of schemes. Additionally, the MPS approach is demonstrated to be a credible and competitive alternative to the conventional MLE method.

\begin{sidewaystable}
\renewcommand{\arraystretch}{1.2}
    \begin{table}[H]
    \smaller
    \centering
    \caption{Point estimates and corresponding MSE for $\eta_1=1.9,\eta_2=2.1,\varpi=3$.}
    \vspace{0.1cm}
    \label{sim point1} 
    \begin{tabular}{cccccccccccccccc}
    \hline 
    && \multicolumn{2}{c}{Classical} & \multicolumn{6}{c}{Informative} & \multicolumn{6}{c}{Non-informative}\\
    \cmidrule{3-4} \cmidrule(l){5-10} \cmidrule(l){11-16}
    && MLE & MPS & \multicolumn{3}{c}{MLE Based} & \multicolumn{3}{c}{MPS Based} & \multicolumn{3}{c}{MLE Based} & \multicolumn{3}{c}{MPS Based}\\
    \cmidrule(l){5-7} \cmidrule(l){8-10} \cmidrule(l){11-13} \cmidrule(l){14-16}
        $(N,n,V,v,s)$ & CS & & & SE & h= 0.5 & h= -0.5 & SE & h= 0.5 & h= -0.5 & SE & h= 0.5 & h= -0.5 & SE & h= 0.5 & h= -0.5\\
        \cmidrule(l){1-16}
         (15,10,10,5,2) & 1 & 0.71164 & 0.72238 & 0.68879 & 0.68677 & 0.69079 & 0.73546 & 0.73360 & 0.73730 & 0.69190 & 0.68934 & 0.69441 & 0.71108 & 0.70854 & 0.71358\\
         & & 0.00787 & 0.01055 & 0.00608 & 0.00616 & 0.00600 & 0.00584 & 0.00574 & 0.00593 & 0.01027 & 0.01039 & 0.01016 & 0.00657 & 0.00658 & 0.00658\\
         & 2 & 0.71542 & 0.70164 & 0.68516 & 0.68308 & 0.68722 & 0.72932 & 0.72745 & 0.73117 & 0.68845 & 0.68587 & 0.69099 & 0.70874 & 0.70622 & 0.71122\\
         & & 0.00752 & 0.01402 & 0.00665 & 0.00675 & 0.00655 & 0.00564 &0.00557 & 0.00571 & 0.01096 & 0.01111 & 0.01083 & 0.00650 & 0.00652 & 0.00650\\
         (15,10,10,5,3) & 1 & 0.54772 & 0.56082 & 0.52889 & 0.52686 & 0.53091 & 0.57695 & 0.57487 & 0.57903 & 0.53579 & 0.53320 & 0.53837 & 0.55312 & 0.55044 & 0.55579\\
         & & 0.00909 & 0.01277 & 0.00647 & 0.00647 & 0.00648 & 0.00753 & 0.00735 & 0.00772 & 0.01146 & 0.01143 & 0.01151 & 0.00795 & 0.00784 & 0.00808\\
         & 2 & 0.55315 & 0.54294 & 0.52744 & 0.52536 & 0.52951 & 0.57273 & 0.57065 & 0.57480 & 0.53396 & 0.53134 & 0.53657 & 0.55278 & 0.55013 & 0.55543\\
         & & 0.00806 & 0.01467 & 0.00622 & 0.00622 & 0.00622 & 0.00674 & 0.00657 & 0.00691 & 0.01107 & 0.01104 & 0.01111 & 0.00720 & 0.00708 & 0.00732\\
         (20,13,12,7,2) & 1 & 0.78612 & 0.79867 & 0.77105 & 0.76963 & 0.77244 & 0.80197 & 0.80070 & 0.80322 & 0.77203 & 0.77028 & 0.77375 & 0.78168 & 0.77995 & 0.78338\\
         & & 0.00517 & 0.00652 & 0.00433 & 0.00439 & 0.00427 & 0.00356 & 0.00354 & 0.00358 & 0.00666 & 0.00676 & 0.00657 & 0.00445 & 0.00450 & 0.00440\\
         & 2 & 0.78324 & 0.78828 & 0.76753 & 0.76610 & 0.76895 & 0.79544 & 0.79415 & 0.79670 & 0.76773 & 0.76596 & 0.76948 & 0.77619 & 0.77446 & 0.77789\\
         & & 0.00470 & 0.00771 & 0.00454 & 0.00462 & 0.00447 & 0.00341 & 0.00341 & 0.00342 & 0.00688 & 0.00699 & 0.00677 & 0.00438 & 0.00445 & 0.00432\\
         (20,13,12,7,3) & 1 & 0.66396 & 0.67800 & 0.65092 & 0.64918 & 0.65265 & 0.68641 & 0.68471 & 0.68809 & 0.65286 & 0.65072 & 0.65499 & 0.66389 & 0.66172 & 0.66605\\
         & & 0.00706 & 0.00918 & 0.00566 & 0.00570 & 0.00562 & 0.00513 & 0.00506 & 0.00521 & 0.00889 & 0.00895 & 0.00883 & 0.00607 & 0.00608 & 0.00607\\
         & 2 & 0.66396 & 0.67109 & 0.64994 & 0.64819 & 0.65169 & 0.68176 & 0.68009 & 0.68341 & 0.65239 & 0.65024 & 0.65452 & 0.66088 & 0.65874 & 0.66299\\
         & & 0.00630 & 0.01057 & 0.00570 & 0.00574 & 0.00565 & 0.00492 & 0.00487 & 0.00498 & 0.00877 & 0.00883 & 0.00872 & 0.00585 & 0.00587 & 0.00585\\
        \cmidrule(l){1-16}
    \end{tabular}
\end{table}
\end{sidewaystable}
\begin{sidewaystable}
\renewcommand{\arraystretch}{1.2}
     \begin{table}[H]
     \smaller
    \centering
    \caption{Interval estimates for $\eta_1=1.9,\eta_2=2.1,\varpi=3$.}
    \vspace{0.1cm}
    \label{sim length1}
    \begin{tabular}{cccccccccccccc}
    \toprule 
    && \multicolumn{6}{c}{Avg. length} & \multicolumn{6}{c}{CP}\\
    \cmidrule(l){3-8} \cmidrule(l){9-14}
    && \multicolumn{2}{c}{Classical} & \multicolumn{2}{c}{Informative} & \multicolumn{2}{c}{Non-informative}& \multicolumn{2}{c}{Classical} & \multicolumn{2}{c}{Informative} & \multicolumn{2}{c}{Non-informative}\\
    \cmidrule(l){3-4} \cmidrule(l){5-6} \cmidrule(l){7-8} \cmidrule(l){9-10} \cmidrule(l){11-12} \cmidrule(l){13-14}
        $(N,n,V,v,s)$ & CS & MLE & MPS & MLE & MPS & MLE & MPS & MLE & MPS & MLE & MPS & MLE & MPS\\ \cmidrule(l){1-14}
        (15,10,10,5,2) & 1 & 0.39916 & 0.25735 & 0.33830 & 0.32168 & 0.37416 & 0.37346 & 0.94600 & 0.75300 & 0.96700 & 0.94700 & 0.92000 & 0.96500\\
        & 2 & 0.39674 & 0.45156 & 0.34353 & 0.32111 & 0.37606 & 0.37013 & 0.94000 & 0.89400 & 0.96000 & 0.94100 & 0.91200 & 0.95500\\
        (15,10,10,5,3) & 1 & 0.41533 & 0.27295 & 0.34547 & 0.35001 & 0.38717 & 0.39477 & 0.95800 & 0.76600 & 0.96400 & 0.95300 & 0.92000 & 0.96100\\
        & 2 & 0.41595 & 0.46684 & 0.34984 & 0.34668 & 0.38911 & 0.39080 & 0.97400 & 0.92300 & 0.97100 & 0.96900 & 0.92900 & 0.97500\\
        (20,13,12,7,2) & 1 & 0.32088 & 0.20013 & 0.27876 & 0.26223 & 0.30506 & 0.30401 & 0.94800 & 0.74400 & 0.95900 & 0.94800 & 0.92900 & 0.96400\\
        & 2 & 0.32320 & 0.36323 & 0.28174 & 0.26363 & 0.30796 & 0.30333 & 0.95300 & 0.91600 & 0.95800 & 0.95700 & 0.92800 & 0.96600\\
        (20,13,12,7,3) & 1 & 0.36916 & 0.23438 & 0.31803 & 0.31253 & 0.34893 & 0.35291 & 0.95100 & 0.76200 & 0.95500 & 0.96300 & 0.92000 & 0.96000\\
        & 2 & 0.36959 & 0.42180 & 0.31942 & 0.30907 & 0.35000 & 0.34802 & 0.96800 & 0.93200 & 0.96500 & 0.96800 & 0.93400 & 0.97200\\
        \cmidrule(l){1-14}
    \end{tabular}
\end{table} 
\end{sidewaystable}     
\begin{sidewaystable}
\renewcommand{\arraystretch}{1.2}
    \begin{table}[H]
    \smaller
    \centering
    \caption{Point estimates and corresponding MSE for $\eta_1=2.25,\eta_2=2.75,\varpi=2$.}
    \vspace{0.1cm}
    \label{sim point2} 
    \begin{tabular}{cccccccccccccccc}
    \hline 
    && \multicolumn{2}{c}{Classical} & \multicolumn{6}{c}{Informative} & \multicolumn{6}{c}{Non-informative}\\
    \cmidrule{3-4} \cmidrule(l){5-10} \cmidrule(l){11-16}
    && MLE & MPS & \multicolumn{3}{c}{MLE Based} & \multicolumn{3}{c}{MPS Based} & \multicolumn{3}{c}{MLE Based} & \multicolumn{3}{c}{MPS Based}\\
    \cmidrule(l){5-7} \cmidrule(l){8-10} \cmidrule(l){11-13} \cmidrule(l){14-16}
        $(N,n,V,v,s)$ & CS & & & SE & h= 0.5 & h= -0.5 & SE & h= 0.5 & h= -0.5 & SE & h= 0.5 & h= -0.5 & SE & h= 0.5 & h= -0.5\\
        \cmidrule(l){1-16}
         (15,10,10,5,2) & 1 & 0.73070 & 0.74501 & 0.71412 & 0.71231 & 0.71592 & 0.76159 & 0.75998 & 0.76318 & 0.71406 & 0.71159 & 0.71650 & 0.72795 & 0.72550 & 0.73036\\
         & & 0.00634 & 0.00901 & 0.00504 & 0.00513 & 0.00496 & 0.00467 & 0.00460 & 0.00474 & 0.00941 & 0.00957 & 0.00925 & 0.00576 & 0.00584 & 0.00570\\
         & 2 & 0.74029 & 0.73359 & 0.71513 & 0.71327 & 0.71696 & 0.75921 & 0.75758 & 0.76081 & 0.71637 & 0.71392 & 0.71879 & 0.73079 & 0.72837 & 0.73317\\
         & & 0.00609 & 0.01272 & 0.00560 & 0.00569 & 0.00551 & 0.00490 & 0.00485 & 0.00496 & 0.01003 & 0.01020 & 0.00987 & 0.00584 & 0.00591 & 0.00579\\
         (15,10,10,5,3) & 1 & 0.57177 & 0.58914 & 0.55777 & 0.55581 & 0.55971 & 0.60917 & 0.60722 & 0.61111 & 0.56352 & 0.56088 & 0.56613 & 0.57605 & 0.57335 & 0.57872\\
         & & 0.00798 & 0.01179 & 0.00565 & 0.00567 & 0.00563 & 0.00691 & 0.00674 & 0.00708 & 0.01102 & 0.01103 & 0.01103 & 0.00731 & 0.00725 & 0.00738\\
         & 2 & 0.57703 & 0.57482 & 0.55635 & 0.55436 & 0.55834 & 0.60287 & 0.60093 & 0.60480 & 0.56206 & 0.55942 & 0.56468 & 0.57447 & 0.57178 & 0.57713\\
         & & 0.00719 & 0.01498 & 0.00585 & 0.00587 & 0.00583 & 0.00642 & 0.00628 & 0.00657 & 0.01113 & 0.01113 & 0.01113 & 0.00682 & 0.00676 & 0.00689\\
         (20,13,12,7,2) & 1 & 0.80611 & 0.82084 & 0.79602 & 0.79480 & 0.79723 & 0.82775 & 0.82669 & 0.82880 & 0.79489 & 0.79328 & 0.79648 & 0.80100 & 0.79939 & 0.80258\\
         & & 0.00436 & 0.00555 & 0.00363 & 0.00369 & 0.00356 & 0.00287 & 0.00286 & 0.00288 & 0.00617 & 0.00629 & 0.00606 & 0.00408 & 0.00416 & 0.00401\\
         & 2 & 0.80542 & 0.81721 & 0.79529 & 0.79406 & 0.79650 & 0.82320 & 0.82213 & 0.82425 & 0.79540 & 0.79379 & 0.79698 & 0.79900 & 0.79742 & 0.80056\\
         & & 0.00403 & 0.00709 & 0.00393 & 0.00400 & 0.00387 & 0.00294 & 0.00294 & 0.00294 & 0.00646 & 0.00658 & 0.00635 & 0.00424 & 0.00433 & 0.00417\\
         (20,13,12,7,3) & 1 & 0.69084 & 0.70947 & 0.68225 & 0.68065 & 0.68384 & 0.71974 & 0.71823 & 0.72124 & 0.68469 & 0.68261 & 0.68675 & 0.68976 & 0.68765 & 0.69184\\
         & & 0.00572 & 0.00786 & 0.00451 & 0.00456 & 0.00447 & 0.00431 & 0.00425 & 0.00438 & 0.00779 & 0.00787 & 0.00772 & 0.00536 & 0.00540 & 0.00532\\
         & 2 & 0.68965 & 0.70647 & 0.68187 & 0.68026 & 0.68347 & 0.71450 & 0.71300 & 0.71598 & 0.68509 & 0.68300 & 0.68715 & 0.68732 & 0.68524 & 0.68937\\
         & & 0.00542 & 0.00995 & 0.00498 & 0.00504 & 0.00494 & 0.00439 & 0.00435 & 0.00444 & 0.00829 & 0.00836 & 0.00822 & 0.00552 & 0.00557 & 0.00548\\
        \cmidrule(l){1-16}
    \end{tabular}
\end{table}
\end{sidewaystable}

\begin{sidewaystable}
\renewcommand{\arraystretch}{1.2}
     \begin{table}[H]
     \smaller
    \centering
    \caption{Interval estimates for $\eta_1=2.25,\eta_2=2.75,\varpi=2$.}
    \vspace{0.1cm}
    \label{sim length2}
    \begin{tabular}{cccccccccccccc}
    \toprule 
    && \multicolumn{6}{c}{Avg. length} & \multicolumn{6}{c}{CP}\\
    \cmidrule(l){3-8} \cmidrule(l){9-14}
    && \multicolumn{2}{c}{Classical} & \multicolumn{2}{c}{Informative} & \multicolumn{2}{c}{Non-informative}& \multicolumn{2}{c}{Classical} & \multicolumn{2}{c}{Informative} & \multicolumn{2}{c}{Non-informative}\\
    \cmidrule(l){3-4} \cmidrule(l){5-6} \cmidrule(l){7-8} \cmidrule(l){9-10} \cmidrule(l){11-12} \cmidrule(l){13-14}
        $(N,n,V,v,s)$ & CS & MLE & MPS & MLE & MPS & MLE & MPS & MLE & MPS & MLE & MPS & MLE & MPS\\ \cmidrule(l){1-14}
        (15,10,10,5,2) & 1 & 0.39253 & 0.25155 & 0.31983 & 0.29815 & 0.36564 & 0.36493 & 0.96500 & 0.78300 & 0.97200 & 0.95400 & 0.92300 & 0.97500\\
        & 2 & 0.38575 & 0.43395 & 0.32348 & 0.29655 & 0.36354 & 0.35925 & 0.95900 & 0.89600 & 0.96300 & 0.94400 & 0.91600 & 0.96200\\
        (15,10,10,5,3) & 1 & 0.41883 & 0.27420 & 0.33883 & 0.33769 & 0.38931 & 0.39482 & 0.97100 & 0.78000 & 0.97900 & 0.94400 & 0.93300 & 0.97600\\
        & 2 & 0.41776 & 0.46908 & 0.34216 & 0.33416 & 0.38924 & 0.39149 & 0.97900 & 0.91000 & 0.97500 & 0.95800 & 0.92500 & 0.97100\\
        (20,13,12,7,2) & 1 & 0.30776 & 0.18928 & 0.25784 & 0.23734 & 0.28895 & 0.29012 & 0.95700 & 0.73800 & 0.96500 & 0.95200 & 0.93100 & 0.97100\\
        & 2 & 0.30800 & 0.33899 & 0.25869 & 0.23722 & 0.28892 & 0.28713 & 0.96300 & 0.88800 & 0.96400 & 0.94300 & 0.91900 & 0.96500\\
        (20,13,12,7,3) & 1 & 0.36449 & 0.22964 & 0.30364 & 0.29390 & 0.34223 & 0.34456 & 0.96600 & 0.78700 & 0.97300 & 0.95500 & 0.92800 & 0.97000\\
        & 2 & 0.36395 & 0.40704 & 0.30472 & 0.29107 & 0.34240 & 0.34121 & 0.96900 & 0.90700 & 0.96400 & 0.95100 & 0.93000 & 0.96500\\
        \cmidrule(l){1-14}
    \end{tabular}
\end{table} 
\end{sidewaystable}

\newpage
\section{Application}\label{application}
We analyze a real data set that is originally reported by \citet{Musa}. \citet{lyu96} also can be referred in this regard. The data are publicly available and can be accessed from the reliability data repository at \url{http://www.cse.cuhk.edu.hk/~lyu/book/reliability/DATA/CH7/SYS2.DAT}. This data set consists of observed failure times of multiple units recorded during a reliability study. In the complete data case, we focus on a 3-out-of-6: G system, corresponding to the choices $s=3$ and $v=6$. This system configuration implies that the system fails when at least three out of six components have failed. To construct the system-level data from the original sequence of failure times, we adopt the following grouping scheme. Let $Y_1$ be the 11th failure time and $X_{1v},v=1,2,\dots,6$ be the failure time observations numbered from 12 to 17. Similarly, $Y_2$ be the 18th failure time while $X_{2v},v=1,2,\dots,6$ denotes the 19th to 24th failure observations. Continuing this process up to the 66th failure, we have obtained data $Y$ with $M=8$. Following the support of the UGR model, we have divided the whole data set by $5000$, and the data are given below:
\begin{equation*}\label{real-data}\underbrace{
    \begin{pmatrix}
        0.2548 & 0.0938 & 0.2348 & 0.1386 & 0.3816 & 0.0270 \\
0.1192 & 0.1514 & 0.0874 & 0.4460 & 0.0874 & 0.0680 \\
0.1070 & 0.0554 & 0.0726 & 0.1044 & 0.1226 & 0.0554 \\
0.1642 & 0.0426 & 0.3240 & 0.3202 & 0.0596 & 0.1748 \\
0.5280 & 0.0010 & 0.0298 & 0.2068 & 0.4882 & 0.0920 \\
0.2238 & 0.0874 & 0.1854 & 0.8924 & 0.1428 & 0.0362 \\
0.1514 & 0.6308 & 0.4230 & 0.1768 & 0.4074 & 0.2962 \\
0.0980 & 0.1186 & 0.3538 & 0.0170 & 0.5672 & 0.0426 \\
    \end{pmatrix} }_{\text{Observed strength data}}
    \quad \text{and} \underbrace{
    \begin{pmatrix}
        0.0234 \\
        0.0554 \\
        0.0810 \\
        0.2600 \\
        0.1236 \\
        0.1130 \\
        0.2970 \\
        0.1118
    \end{pmatrix}}_{\text{Observed stress data}}
\end{equation*}

We first examine whether our model is suitable for analyzing the given dataset. Table 5 presents the maximum likelihood estimates of the unknown parameters for the three competing models —unit generalized Rayleigh, unit Gompertz, and Kumarswamy and unit inverse Weibull—based on both data sets, along with the corresponding Kolmogorov–Smirnov (K–S) statistics and p-values. From these results, it is evident that the unit generalized Rayleigh distribution yields smaller K–S statistics and larger p-values than the alternative models, indicating a better overall fit to the observed data and supporting its suitability for the present analysis. Figure~\ref{fit} presents additional graphical results, where we overlay the empirical cumulative distribution functions with the corresponding theoretical CDFs, along with P–P plots, Q–Q plots, and histograms fitted with PDFs using both MLE and MPS methods. These plots demonstrate a good model fit. To examine parameter convergence in the Bayesian framework, we include trace plots with posterior means (blue dotted lines) and $95\%$ HPD intervals (blue solid lines), as well as posterior histograms and mean convergence plots. Figures~\ref{scheme1} and \ref{scheme2} clearly confirm stable parameter behavior, allowing reliable sampling from the posterior distributions.

We present the reliability estimates obtained from the complete data using the MLE, MPS, and the Bayesian estimates, along with their associated $95\%$ ACIs and HPDs. The Bayesian estimates are computed under noninformative prior distributions to avoid the influence of subjective prior information. In addition, reliability estimates $\Psi_{s,v}$ are derived under PT2 censoring by considering two distinct CSs, allowing us to examine the effect of censoring on the estimation procedure. The corresponding estimates and interval results for these schemes are presented in the Table \ref{real data est}. We consider the following schemes:
\begin{enumerate}[Scheme 1:]
    \item $R = (1,0,0,0,0), S = (2,0,0,0,0,0), (V = 6, v = 5, M = 8, m = 6, s = 3)$
    \item $R = (0,0,0,0,1), S = (0,0,0,0,0,2), (V = 6, v = 5, M = 8, m = 6, s = 3)$
\end{enumerate}

Based on these two CSs, the generated PT2 censored samples are given below:
\begin{equation*}\label{real-data-cen1}
\underbrace{\underbrace{
    \begin{pmatrix}
        0.0874 & 0.1226 & 0.1854 & 0.2548 & 0.5280 \\
0.0298 & 0.0874 & 0.0980 & 0.1428 & 0.2348 \\
0.1070 & 0.1514 & 0.3240 & 0.3816 & 0.4882 \\
0.0170 & 0.0270 & 0.0554 & 0.3202 & 0.6308 \\
0.0596 & 0.0726 & 0.1192 & 0.2238 & 0.4230 \\
0.0426 & 0.0874 & 0.1514 & 0.1748 & 0.1768 \\
    \end{pmatrix} }_{\text{Observed strength data}}
    \quad \text{and} \underbrace{
    \begin{pmatrix}
0.0234\\
0.0810\\
0.1130\\
0.1236\\
0.2600\\
0.2970
    \end{pmatrix}}_{\text{Observed stress data}}}_{\text{Under Scheme 1}}
\end{equation*}
\begin{equation*}\label{real-data-cen2}
\underbrace{\underbrace{
    \begin{pmatrix}
        0.0874 & 0.1226 & 0.1854 & 0.2548 & 0.4074 \\
0.0010 & 0.0554 & 0.0938 & 0.2962 & 0.4460 \\
0.0298 & 0.0874 & 0.0980 & 0.1428 & 0.1642 \\
0.0362 & 0.0426 & 0.0680 & 0.1186 & 0.1386 \\
0.1070 & 0.1514 & 0.3240 & 0.3538 & 0.3816 \\
0.0170 & 0.0270 & 0.0554 & 0.0920 & 0.3202 \\
    \end{pmatrix} }_{\text{Observed strength data}}
    \quad \text{and} \underbrace{
    \begin{pmatrix}
0.0234\\
0.0554\\
0.0810\\
0.1118\\
0.1130\\
0.1236
    \end{pmatrix}}_{\text{Observed stress data}}}_{\text{Under Scheme 2}}
\end{equation*}

\begin{figure}[H]
    \centering
    \includegraphics[width=1\linewidth]{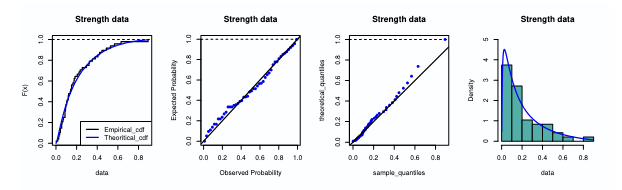}
    \includegraphics[width=1\linewidth]{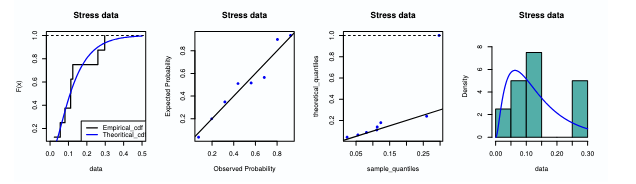}
    \caption{Empirical distribution with theoretical distribution, P-P, Q-Q, and histogram with fitted density plot for real data.}
    \label{fit}
\end{figure}

\begin{table}[H]
    \centering
    \smaller
    \caption{Goodness of fit for the real data.}
    \vspace{0.1cm}
    \label{goodness-of-fit} 
    \begin{tabular}{ccccccccc}
    \hline 
    Distribution & \multicolumn{4}{c}{Strength data} & \multicolumn{4}{c}{Stress data}\\
    \cmidrule(l){2-5} \cmidrule(l){6-9}
    & $\hat{\eta_1}$ & $\hat{\varpi}$ & K-S & p-value & $\hat{\eta_2}$ & $\hat{\varpi}$ & K-S & p-value\\
    \cmidrule{1-9}
    Unit Generalized Rayleigh & 0.9362 & 0.4180 & 0.0858 & 0.8711 & 2.6758 & 0.5499 & 0.1862 & 0.8999\\
     \cmidrule{1-9}
                Unit Gompertz & 0.7378 & 0.3645 & 0.1189 & 0.5057 & 0.0523 & 1.1408 & 0.1823 & 0.9124\\
     \cmidrule{1-9}
                   Kumarswamy & 3.2295 & 0.9199 & 0.1089 & 0.6209 & 16.4038 & 1.4908 & 0.2256 & 0.7326\\
     \cmidrule{1-9}
          Unit Inverse Weibull & 1.1866 & 1.0799 & 0.2149 & 0.0238 & 5.6605 & 2.9644 & 0.2795 & 0.4772\\
      \cmidrule{1-9}
    \end{tabular}
\end{table}

\begin{table}[H]
    \centering
    \smaller
    \caption{Point and interval estimates of $\Psi_{s,v}$ for the real data.}
    \vspace{0.1 cm}
    \label{real data est} 
    \begin{tabular}{cccccccc}
    \hline 
    & & & \multicolumn{3}{c}{Non-informative} & \\
    \cmidrule(l){4-6}
    CS & Method & Classical & SE & h= -0.5 & h= 0.5 & ACI & HPD\\
    \hline
    Complete & MLE & 0.73067 & 0.68137 & 0.67842 & 0.68427 & [0.50498, 0.95636] & [0.47043, 0.87020]\\
             & MPS & 0.70184 & 0.66744 & 0.66472 & 0.67009 & [0.53790, 0.86579] & [0.46315, 0.85101]\\
    \cmidrule{1-8}
    Scheme 1 & MLE & 0.61069 & 0.56025 & 0.55663 & 0.56381 & [0.33381, 0.88757] & [0.31927, 0.77643]\\
             & MPS & 0.59309 & 0.55309 & 0.54950 & 0.55664 & [0.39490, 0.79127] & [0.31278, 0.76335]\\
    \cmidrule{1-8}         
    Scheme 2 & MLE & 0.66932 & 0.61646 & 0.61282 & 0.62003 & [0.40142, 0.93723] & [0.37765, 0.82973]\\
             & MPS & 0.62222 & 0.57068 & 0.56706 & 0.57423 & [0.36881, 0.87563] & [0.33565, 0.78311]\\
    \cmidrule{1-8}
    \end{tabular}
\end{table}

Based on the Table \ref{real data est}, we observe that MLE and MPS methods provide similar estimates of the unknown parametric function. The HPD credible intervals based on the non-informative prior are slightly smaller than the corresponding length of the ACI estimates. Additionally, we can conclude that the MPS method yields a shorter interval length than the MLE method, and this holds true for both classical and Bayesian cases.

\begin{figure}
    \centering
    \includegraphics[width=0.48\linewidth]{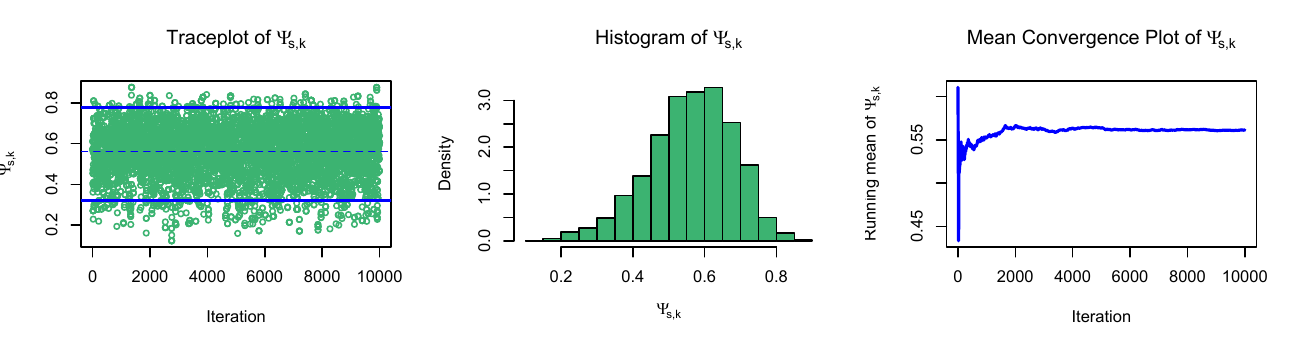}
    \includegraphics[width=0.48\linewidth]{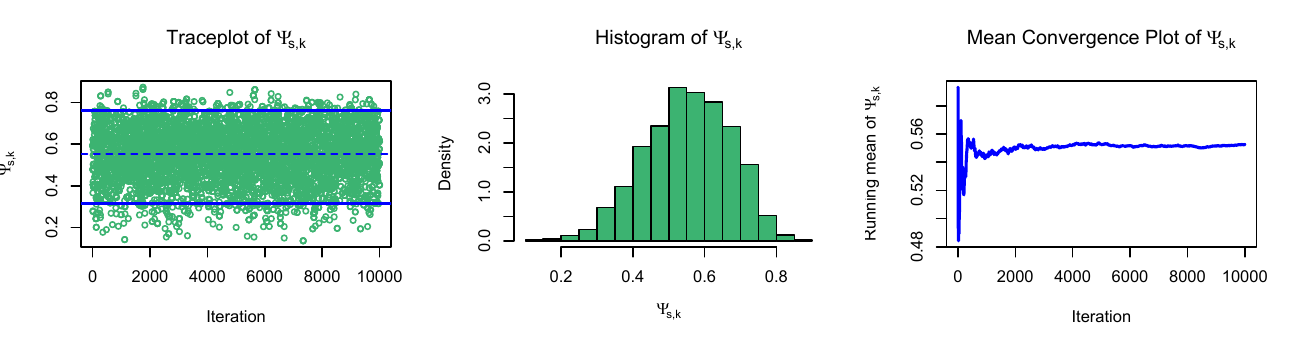}
    \caption{MCMC trace, histogram, and mean convergence plot of real data under Scheme~1 based on MLE (upper panel) and MPS (lower panel).}
    \label{scheme1}
\end{figure}
\begin{figure}
    \centering
    \includegraphics[width=0.48\linewidth]{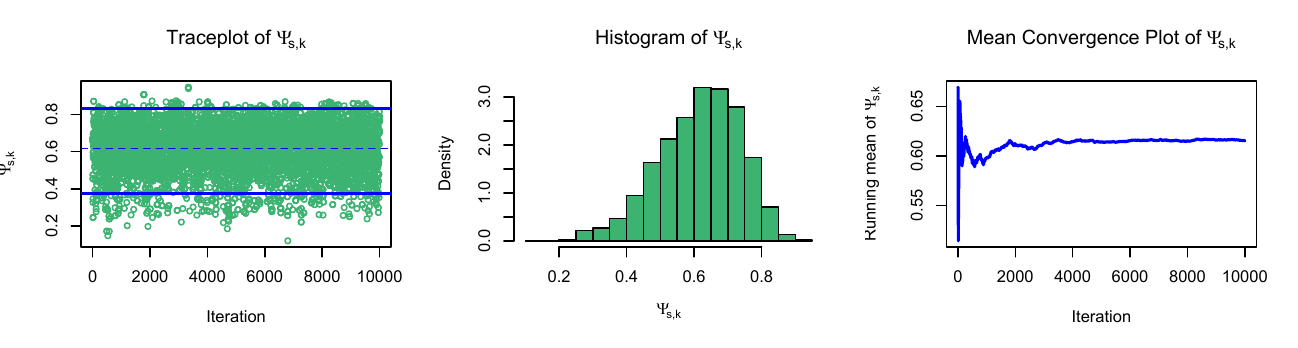}
    \includegraphics[width=0.48\linewidth]{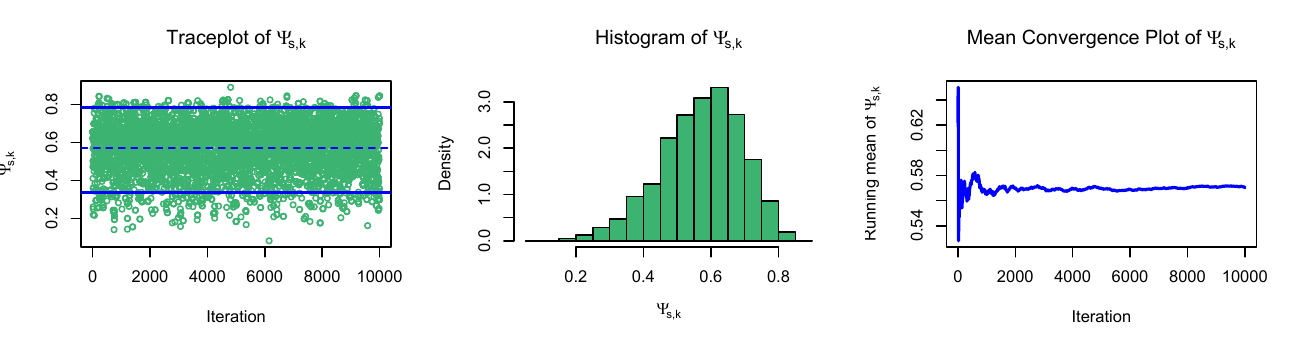}
    \caption{MCMC trace, histogram, and mean plot of real data under Scheme~2 based on MLE (upper panel) and MPS (lower panel).}
    \label{scheme2}
\end{figure}

\section{Optimal censoring plan}\label{optimal cs}
Previous sections focus on parameter estimation under a given PT2 censoring scheme, assuming that item lifetimes follow the UGR distribution and that the CSS $\mathcal{R}$ is known a priori. In practice, however, this assumption is often limiting, 
as the censoring pattern is typically determined by the experimenter. For a fixed sample size $n$ and $m$ observed failures, there exist $\binom{n-1}{m-1}$ possible PT2 censoring 
schemes. This number grows rapidly; for instance, when $n=35$ and $m=25$, there are $\binom{34}{24}=131{,}128{,}140$ admissible schemes, rendering exhaustive evaluation 
infeasible.

This large variety of options naturally leads to the question: which CS should one choose? Since different schemes may yield significantly different levels of information about the unknown parameters, designing an optimal censoring plan becomes a crucial task. Rather than simply taking the first feasible scheme, one can deliberately search for a scheme that performs best according to a specific criterion. Depending on the context, this may involve maximizing the Fisher information, reducing the variance of the estimators, or achieving more accurate predictions for remaining lifetimes. In this way, the CS becomes a design variable that can be optimized, rather than a fixed assumption. The goal of this section is to develop such optimal censoring strategies, providing practical guidelines for planning more informative and efficient life-testing experiments under PT2 censoring. 
For fixed values of $n$ and $m$, we denote the total CS, set as $CS(n,m)$, and all possible schemes can be formally defined as
$$CS(n,m)=\left\{\mathcal{R}=(r_2,r_2,\dots,r_m)\in\mathcal{N}^m \mid \sum_{i=1}^mr_i=n-m,\mathcal{N}=\{0,1,2,\dots,n-m\} \right\}.$$
Even for moderate values of $n$ and $m$, the choice of an appropriate CS constitutes a key element in the design of an efficient life-testing experiment. A judiciously selected scheme may enhance the accuracy of parameter estimation, reduce experimental costs, and reinforce the robustness of resulting inferences. Owing to its practical relevance, the task of determining an optimal censoring design has attracted substantial interest in the literature. Consequently, a variety of criteria and methodological frameworks have been advanced to identify CSs that perform favorably under diverse objectives and operational constraints (see, for instance, \citet{Pradhan1}, \citet{Chandra}, \citet{Alotaibi}, 
\citet{Sen}, among others).

\subsection{Optimality criteria}
This section discusses various optimality criteria for determining the optimal CS that provides the maximum information about the unknown parameters. Three criteria are taken into consideration when performing optimization.  It should be noted that the exact Fisher information matrix, $\mathcal I_E(\Theta)$, which is obtained in Section \ref{EFIM}, serves as the foundation for the objective functions defined here.
 
Our first considered design criterion focuses on the overall precision of the model parameter estimators. This is achieved by using the determinant of the inverse of the Fisher information matrix. Therefore, we define our second objective function $\phi_1(\mathcal{R})$ as $$\phi_1(\mathcal{R}) = \det\left(\mathcal I^{-1}_E(\Theta)\right).$$
Accordingly, the second optimal design $\mathcal{R}^*$ is obtained by
$$\mathcal{R}^* = \arg\min_{\mathcal{R}} , \phi_1(\mathcal{R}).$$

Another optimization criterion explored in this article is the trace of the first-order approximation of the variance-covariance matrix based on the MLEs. The trace corresponds to the sum of the diagonal elements of the inverse of FIM.  This optimality criterion calculates the variance of parameter estimations by summing the eigenvalues of the inverse of the Fisher information matrix.  Consequently, we achieve the optimal design value $\mathcal{R}^*$ by minimizing the objective function $\phi_2(\mathcal{R})=tr\left(\mathcal I^{-1}_E(\Theta)\right)$
as $$\mathcal{R}^*=\arg\min_{\mathcal{R}}\phi_2(\mathcal{R}).$$

Reliability prediction plays a central role in both product design and the testing phase of development. To obtain a dependable estimate of a product’s reliability, we adopt a test design criterion that aims to minimize the asymptotic variance (AV) of the reliability estimate of the multicomponent system. We obtain the asymptotic variance (AV) of the reliability using the delta method. The resulting expression is given by:
\[AV( \Psi_{s,v}) = \nabla^{\top}\mathcal I^{-1}_E(\Theta)\nabla,\]
where elements of $\nabla$ form a column vector, obtained by differentiating the reliability function with respect to each model parameter. In particular, the column vector $\nabla$ is expressed as:
$$\nabla=\left(\frac{\partial\Psi_{s,v}}{\partial\eta_1},\frac{\partial\Psi_{s,v}}{\partial\eta_2},\frac{\partial\Psi_{s,v}}{\partial\varpi} \right)^{\top}.$$
We then define our third design objective function $\phi_3(\mathcal{R})$ as the asymptotic variance of the reliability under design $\mathcal{R}$:
$$\phi_3(\mathcal{R}) = AV( \Psi_{s,v}).$$

Hence, the optimal CS $\mathcal{R}^*$ can be obtained by solving
\[\mathcal{R}^*=\arg \min_{\mathcal{R}} \phi_3(\mathcal{R}).\]

To see how the method works in practice, we set out to find the best CS together with its optimal objective value. We used the parameter settings as $(\eta_1,\eta_2,\varpi)=(1.2,1.5,2)$. For each chosen sample size, we generated all possible schemes and evaluated them using the proposed objective function. The scheme that produced the smallest value is selected as the optimal one. The final set of optimal schemes, along with their corresponding objective values, is displayed in the Table \ref{optimality} for arbitrary $s=3$. These results demonstrate how the preferred scheme evolves as the sample size varies, providing a practical guide for selecting efficient censoring strategies in real-world applications. From Table \ref{optimality}, it is clear that the optimal CS varies with the chosen optimality criterion, indicating that no single scheme is universally optimal given that the systems and the corresponding components for the individual system are provided beforehand. Moreover, the objective values $\phi(\mathcal{R}^*)$ decrease as the sample sizes increase, reflecting improved estimation efficiency. Overall, these results demonstrate that appropriately chosen PT2 censoring schemes can substantially enhance the precision of reliability estimation.

\begin{table}[H]
\centering
\renewcommand{\arraystretch}{1.2} 
\setlength{\tabcolsep}{4pt}    
\caption{Optimal solution and corresponding performance values with $s=3$.}
\label{optimality}
\begin{tabular}{cccc} 
\hline
$(N, n, V, v)$ & Criteria & $\mathcal{R}^*$ & $\phi(\mathcal{R}^*)$ \\ 
\hline
(10, 6, 8, 5) 
& $\phi_1$ & $\mathbf{R} = (0,0,0,2,1),\ \mathbf{S} = (0,0,0,0,0,4)$ & 0.478605 \\
& $\phi_2$ & $\mathbf{R} = (0,0,1,2,0),\ \mathbf{S} = (0,0,0,0,3,1)$ & 0.001001 \\
& $\phi_3$ & $\mathbf{R} = (0,0,3,0,0),\ \mathbf{S} = (1,0,3,0,0,0)$ & 0.019850 \\ 
\hline
(14, 10, 8, 5)
& $\phi_1$ & $\mathbf{R} = (0,0,0,2,1),\ \mathbf{S} = (0,0,0,0,0,0,4,0,0,0)$ & 0.287158 \\
& $\phi_2$ & $\mathbf{R} = (0,0,1,2,0),\ \mathbf{S} = (0,0,0,0,0,3,1,0,0,0)$ & 0.000216 \\
& $\phi_3$ & $\mathbf{R} = (0,0,0,0,3),\ \mathbf{S} = (4,0,0,0,0,0,0,0,0,0)$ & 0.011910 \\ 
\hline
(14, 10, 10, 7)
& $\phi_1$ & $\mathbf{R} = (0,0,0,2,1,0,0),\ \mathbf{S} = (0,0,0,0,0,0,4,0,0,0)$ & 0.270581 \\
& $\phi_2$ & $\mathbf{R} = (0,0,0,2,1,0,0),\ \mathbf{S} = (0,0,0,0,0,3,1,0,0,0)$ & 0.000116 \\
& $\phi_3$ & $\mathbf{R} = (1,0,0,0,1,1,0),\ \mathbf{S} = (4,0,0,0,0,0,0,0,0,0)$ & 0.011636 \\ 
\hline
(18, 14, 10, 7)
& $\phi_1$ & $\mathbf{R} = (0,0,0,2,1,0,0),\ \mathbf{S} = (0,0,0,0,0,0,0,4,0,0,0,0,0,0)$ & 0.193272 \\
& $\phi_2$ & $\mathbf{R} = (0,0,0,2,1,0,0),\ \mathbf{S} = (0,0,0,0,0,0,0,4,0,0,0,0,0,0)$ & 0.000042 \\
& $\phi_3$ & $\mathbf{R} = (0,0,0,0,0,0,3),\ \mathbf{S} = (4,0,0,0,0,0,0,0,0,0,0,0,0,0)$ & 0.008311 \\ 
\hline
\end{tabular}
\end{table}

\section{Conclusion}\label{conclusion}
In this article, we investigate the estimation of multicomponent stress–strength reliability under a PT2 censoring scheme using the unit generalized Rayleigh distribution. Maximum likelihood estimates of the three unknown model parameters, along with the associated reliability measure $\Psi_{s,v}$, are obtained via the EM algorithm. The ACI of the reliability is derived from the FIM using the missing-information principle, and the results demonstrate satisfactory coverage properties even for relatively small sample sizes. As an alternative to maximum likelihood estimation, we also propose the maximum product spacing method and compare its performance with that of the MLE. From a Bayesian perspective, point estimates of the reliability function and corresponding credible intervals are obtained using MCMC techniques under both informative and non-informative prior distributions. A comprehensive Monte Carlo simulation study is conducted to assess and compare the performance of the proposed estimators. The simulation results consistently indicate that the MPS estimator outperforms the MLE in terms of accuracy and stability, a conclusion that holds under both classical and Bayesian frameworks. All simulations are implemented using the \textbf{R} programming environment, employing the \textbf{``nleqslv''} package for classical inference and the \textbf{``mcmc''} and \textbf{``coda''} packages for Bayesian convergence diagnostics. To illustrate the practical applicability of the proposed methodology, a real software stress–strength reliability data set is analyzed. In addition, optimal PT2 censoring schemes are identified under three different optimality criteria, providing valuable guidance for experimental design. Finally, several promising directions for future research are identified, including extensions to more advanced CSs, accelerated life testing environments, and alternative sampling strategies such as ranked set sampling.

\appendix

\section*{Appendix 1}\label{appendix-1}
\subsection*{Proof of the Lemma \ref{lemma1}}
\begin{proof}
      The conditional expectations are obtained by using the result that, conditional on $x_i$, the random variable $z_{x_i}$ follows a left-truncated distribution $F$ with truncation point at $x_i$. Consequently, the conditional distribution of $z_{x_i}$ is derived as

    $$f(z_{x_i}\mid x_i;\eta,\varpi)=\frac{f(z_{x_i};\eta,\varpi)}{1-F(x_i;\eta,\varpi)};z_{x_i}>x_i,i=1,2,\dots,n,$$
    where $F$ and $f$ denote CDF and PDF of the $UGR(\eta,\varpi)$ distribution respectively.
    
    Now, let \(g(z_{x_i})=\left(\ln z_{x_i}\right)^2\), then
\begin{equation*}
    \begin{aligned}
        E[g(z_{x_i}) \mid z_{x_i} > x_i] = & \frac{1}{1 - F(x_i)} \int_{x_i}^{1} g(z) f(z;\eta,\varpi)dz\\
        = & \frac{2\eta\varpi^2}{1 - F(x_i)} \int_{x_i}^{1} \left(\ln z\right)^2 \left(-\frac{\ln z}{z}\right)e^{-(\varpi\ln z)^2}\left(1-e^{-(\varpi\ln z)^2}\right)^{\eta-1}dz.
    \end{aligned}
\end{equation*}
Using the transformation $u = 1 - e^{-(\varpi \ln z)^2}$, we have the stated results (i)
\begin{equation*}
    E\left[\left(\ln z_{x_i}\right)^2 \mid z_{x_i} > x_{i}\right] = -\frac{\eta}{\varpi^2(1-F(x_{i}))}\int_0^{1-e^{-(\varpi\ln x_{i})^2}} \ln(1-u)u^{\eta-1}du.
\end{equation*}
In a similar way, let \(h(z_{x_i})=\ln\left(1-e^{-(\varpi\ln z_{x_i})^2}\right)\), then
\begin{equation*}
    \begin{aligned}
        E[h(z_{x_i}) \mid z_{x_i} > x_i] = & \frac{1}{1 - F(x_i)} \int_{x_i}^{1} h(z) f(z;\eta,\varpi)dz\\
        = & \frac{2\eta\varpi^2}{1 - F(x_i)} \int_{x_i}^{1} \ln\left(1-e^{-(\varpi\ln z)^2}\right) \left(-\frac{\ln z}{z}\right)e^{-(\varpi\ln z)^2}\left(1-e^{-(\varpi\ln z)^2}\right)^{\eta-1}dz.
    \end{aligned}
\end{equation*}
Now, using the transformation $u = 1 - e^{-(\varpi \ln z)^2}$, and simplifying the integral, we have the desired result (ii)
\begin{equation*}
    E\left[\ln\left(1-e^{-(\varpi\ln z_{x_i})^2}\right) \mid z_{x_i} > x_{i}\right] = \ln\left(1-e^{-(\varpi\ln x_{i})^2}\right)-\frac{1}{\eta}.
\end{equation*}
\end{proof}

\subsection*{Proof of the Lemma \ref{lemma2}}
\begin{proof}
     Consider a PT2 censored sample $X_1, X_2,\dots, X_n$ from $UGR(\eta,\varpi)$ with CS $(N,n,R_1,R_2,\dots,R_n)$. Then the probability density function of $X_{(i)}$ for $i=1,2,\dots,n$ is 
    \begin{equation}
        \begin{aligned}
            f_{X_{(j)}}(x) & = C_{i-1}\sum_{l=1}^ia_{i,l}f_X(x)[1-F_X(x)]^{\gamma_l-1}\\
            & = C_{i-1}\sum_{l=1}^ia_{i,l}2\eta\varpi^2\left(-\frac{\ln x}{x}\right)e^{-(\varpi\ln x)^2}\left(1-e^{-(\varpi\ln x)^2}\right)^{\eta\gamma_l-1},
        \end{aligned}
    \end{equation}
    where $$\gamma_l=v-l+1+\sum_{l=1}^vR_l, ~C_{i-1}=\prod_{l=1}^i\gamma_l, ~a_{i,l}=\prod_{q=1,q\neq l}^i\frac{1}{\gamma_q-\gamma_l}, ~\text{with}~ a_{1,1}=1.$$

    \begin{enumerate}[(i)]
        \item \begin{align*}
            E\left[ \left(\ln x\right)^2 \right]  = & \int_0^1\left(\ln x\right)^2C_{i-1}\sum_{l=1}^ia_{i,l}2\eta\varpi^2\left(-\frac{\ln x}{x}\right)e^{-(\varpi\ln x)^2}\left(1-e^{-(\varpi\ln x)^2}\right)^{\eta\gamma_l-1}dx\\
            &\text{$\bigg[$Now, using the transformation $u = 1 - e^{-(\varpi \ln x)^2}$$\bigg]$}\\
            = & -\frac{\eta C_{i-1}}{\varpi^2}\sum_{l=1}^ia_{i,l}\int_0^1u^{\eta\gamma_l-1}\ln(1-u)du\\
            = & -\frac{\eta C_{i-1}}{\varpi^2}\sum_{l=1}^ia_{i,l}B(\eta\gamma_l,1)[\psi(1)-\psi(1+\eta\gamma_l)].
        \end{align*}

        \item \begin{align*}
            E\left[\frac{e^{-(\varpi\ln x)^2}\left(\ln x_i\right)^2}{1-e^{-(\varpi\ln x)^2}}\right]  = &\int_0^1 \frac{e^{-(\varpi\ln x)^2}\left(\ln x\right)^2}{1-e^{-(\varpi\ln x)^2}}C_{i-1}\sum_{l=1}^ia_{i,l}2\eta\varpi^2\left(-\frac{\ln x}{x}\right)e^{-(\varpi\ln x)^2}\\
            & \times \left(1-e^{-(\varpi\ln x)^2}\right)^{\eta\gamma_l-1}dx\\
            = & -\frac{\eta C_{i-1}}{\varpi^2}\sum_{l=1}^ia_{i,l}\int_0^1u^{\eta\gamma_l-2}(1-u)\ln(1-u)du\\
            = & \frac{\eta C_{i-1}}{\varpi^2}\sum_{l=1}^ia_{i,l}B(\eta\gamma_l-1,2)[\psi(2)-\psi(1+\eta\gamma_l)].
        \end{align*}

        \item Similarly, we have,
        \begin{align*}
            E&\left[\frac{e^{-(\varpi\ln x)^2}\left(\ln x\right)^2\left(1-e^{-(\varpi\ln x)^2}-2\varpi^2\left(\ln x\right)^2\right)}{\left(1-e^{-(\varpi\ln x)^2}\right)^2}\right]\\
            =&\int_0^1 \frac{e^{-(\varpi\ln x)^2}\left(\ln x\right)^2\left(1-e^{-(\varpi\ln )^2}-2\varpi^2\left(\ln x\right)^2\right)}{\left(1-e^{-(\varpi\ln x)^2}\right)^2} C_{i-1}\sum_{l=1}^ia_{i,l}2\eta\varpi^2\left(-\frac{\ln x}{x}\right)e^{-(\varpi\ln x)^2}\\
            & \times \left(1-e^{-(\varpi\ln x)^2}\right)^{\eta\gamma_l-1}dx\\
            = & -\frac{\eta C_{i-1}}{\varpi^2}\sum_{l=1}^ia_{i,l}\int_0^1 u^{\eta\gamma_l-3}(1-u)\ln(1-u)(u+2\ln(1-u))du\\
            = & -\frac{\eta C_{i-1}}{\varpi^2}\sum_{l=1}^ia_{i,l}\bigg[B(\eta\gamma_l-1,2)[\psi(2)-\psi(\eta\gamma_l+1)] \\
            & + B(\eta\gamma_l-2,2)[(\psi(2)-\psi(\eta\gamma_l))^2 + \psi^{'}(2)-\psi^{'}(\eta\gamma_l)]\bigg].
        \end{align*}
    \end{enumerate}
\end{proof}

\subsection*{Proof of the Theorem \ref{theorem1}}
\begin{proof}
    Let, $\hat{\Theta}^{EM} = (\hat{\eta}_1, \hat{\eta}_2,\varpi)$ is the MLE of the parameter vector $\Theta = (\eta_1, \eta_2)$, then under standard regularity conditions, it holds that
\[
\left(\hat{\Theta}^{EM} - \Theta\right) \xrightarrow{D} N\left(0, \mathcal{I}^{-1}(\hat\Theta^{EM})\right),
\]
where $\mathcal{I}^{-1}(\hat\Theta^{EM})$ is the inverse of the observed Fisher information matrix.

Since $\Psi_{s,v}$ is a smooth and continuously differentiable function of $\Theta$, the Delta method implies that
\[
\left(\hat{\Psi}_{s,v}^{EM} - \Psi_{s,v}\right)
\xrightarrow{D} N\left(0, \mathcal{V}_{\hat \Psi_{s,v}^{EM}}\right),
\]
where $\mathcal{V}_{\hat \Psi_{s,v}^{EM}} = \nabla\Psi_{s,v}^{\top}\mathcal{I}^{-1}(\hat\Theta^{EM})\nabla\Psi_{s,v}$ and 
$\nabla\Psi_{s,v} = \left(\frac{\partial\Psi_{s,v}}{\partial\eta_1},\frac{\partial\Psi_{s,v}}{\partial\eta_2},\frac{\partial\Psi_{s,v}}{\partial\varpi} \right)^{\top}.$

Now, expanding this quadratic form gives
\[\mathcal{V}_{\hat \Psi_{s,v}^{EM}}=\left(\frac{\partial\Psi_{s,v}}{\partial\eta_1}\right)^2\mathcal{I}_{11}^{-1}(\hat\Theta^{EM}) + \left(\frac{\partial\Psi_{s,v}}{\partial\eta_2}\right)^2\mathcal{I}_{22}^{-1}(\hat\Theta^{EM}) + 2\left(\frac{\partial\Psi_{s,v}}{\partial\eta_1}\right)\left(\frac{\partial\Psi_{s,v}}{\partial\eta_2}\right)\mathcal{I}_{12}^{-1}(\hat\Theta^{EM}).\]
And this completes the proof.
\end{proof}

\section*{Appendix 2}\label{appendix-2}
\subsection*{The conditional expectations for the Sect. \ref{mle-em}}
\begin{equation}
    E\left[\left(\ln z_{ijl}\right)^2 \mid z_{ijl} > x_{ij}\right] = -\frac{\eta_1}{\varpi^2(1-F(x_{ij}))}\int_0^{1-e^{-(\varpi\ln x_{ij})^2}} \ln(1-u)u^{\eta_1-1}du = A_1(x_{ij};\eta_1,\varpi),
\end{equation}
\begin{equation}
    E\left[\left(\ln w_{iu}\right)^2 \mid w_{iu} > y_{i}\right] = -\frac{\eta_2}{\varpi^2(1-F(y_{i}))}\int_0^{1-e^{-(\varpi\ln y_{i})^2}} \ln(1-u)u^{\eta_2-1}du = A_2(y_i;\eta_2,\varpi),
\end{equation}
\begin{equation}
    E\left[\ln\left(1-e^{-(\varpi\ln z_{ijl})^2}\right) \mid z_{ijl} > x_{ij}\right] = \ln\left(1-e^{-(\varpi\ln x_{ij})^2}\right)-\frac{1}{\eta_1} = B_1(x_{ij};\eta_1,\varpi),
\end{equation}
and
\begin{equation}
    E\left[\ln\left(1-e^{-(\varpi\ln w_{iu})^2}\right) \mid w_{iu} > y_{i}\right] = \ln\left(1-e^{-(\varpi\ln y_{i})^2}\right)-\frac{1}{\eta_2} = B_2(y_i;\eta_2,\varpi).
\end{equation}

\subsection*{Matrix elements for the Sect \ref{OFIM}}
\begin{equation*}
    \begin{aligned}
        a_{11}= & \frac{Vn}{\eta_1^2},\quad a_{22}=\frac{N}{\eta_2^2},\quad a_{12}=0,\\
        a_{13}= & \frac{2Vn\eta_1}{\varpi}\int_0^1u^{\eta_1-2}(1-u)\ln(1-u)du,\quad a_{23}=\frac{2N\eta_2}{\varpi}\int_0^1u^{\eta_2-2}(1-u)\ln(1-u)du,\\
        a_{33}= & \frac{2(Vn+N)}{\varpi^2} -\frac{2Vn\eta_1}{\varpi^2}\int_0^1u^{\eta_1-1}\ln(1-u)du - \frac{2N\eta_2}{\varpi^2}\int_0^1u^{\eta_2-1}\ln(1-u)du \\ 
            & + \frac{2Vn\eta_1(\eta_1-1)}{\varpi^2}\int_0^1u^{\eta_1-3}(1-u)\ln(1-u)(u+2\ln(1-u))du \\
            & + \frac{2N\eta_2(\eta_2-1)}{\varpi^2}\int_0^1u^{\eta_2-3}(1-u)\ln(1-u)(u+2\ln(1-u))du.
    \end{aligned}
\end{equation*}

\begin{equation*}
    \begin{aligned}
        b_{11}= & \sum_{i=1}^n\sum_{j=1}^v\frac{R_j}{\eta_1^2}, \quad b_{22}= \sum_{i=1}^n\frac{S_i}{\eta_2^2}, \quad b_{12}=0,\\
        b_{13}= & \sum_{i=1}^n\sum_{j=1}^vR_j \left[2\varpi\frac{e^{-(\varpi\ln x_{ij})^2}\left(\ln x_{ij}\right)^2}{1-e^{-(\varpi\ln x_{ij})^2}} \right.\\
        & \left. + \frac{2\eta_1}{\varpi\left(1-e^{-(\varpi\ln x_{ij})^2}\right)^{\eta_1}}\int_0^{1-e^{-(\varpi\ln x_{ij})^2}}u^{\eta_1-2}(1-u)\ln(1-u)du \right],\\
        b_{23} = & \sum_{i=1}^nS_i \left[2\varpi\frac{e^{-(\varpi\ln y_{i})^2}\left(\ln y_{i}\right)^2}{1-e^{-(\varpi\ln y_{i})^2}} + \frac{2\eta_2}{\varpi\left(1-e^{-(\varpi\ln y_{i})^2}\right)^{\eta_2}}\int_0^{1-e^{-(\varpi\ln y_{i})^2}}u^{\eta_2-2}(1-u)\ln(1-u)du \right],\\
        b_{33} = & \sum_{i=1}^n\sum_{j=1}^vR_j\frac{2}{\varpi^2} + \sum_{i=1}^nS_i\frac{2}{\varpi^2} + \sum_{i=1}^n\sum_{j=1}^v2\eta_1R_v \frac{e^{-(\varpi\ln x_{ij})^2}\left(\ln x_{ij}\right)^2\left(1-e^{-(\varpi\ln x_{ij})^2}-2\varpi^2\left(\ln x_{ij}\right)^2\right)}{\left(1-e^{-(\varpi\ln x_{ij})^2}\right)^2} \\
        & + \sum_{i=1}^n2\eta_2S_i \frac{e^{-(\varpi\ln y_{i})^2}\left(\ln y_{i}\right)^2\left(1-e^{-(\varpi\ln y_{i})^2}-2\varpi^2\left(\ln y_{i}\right)^2\right)}{\left(1-e^{-(\varpi\ln y_i)^2}\right)^2} \\
         & - \sum_{i=1}^n\sum_{j=1}^v\frac{2\eta_1}{\varpi^2\left(1-e^{-(\varpi\ln x_{ij})^2}\right)^{\eta_1}}\int_0^{1-e^{-(\varpi\ln x_{ij})^2}}u^{\eta_1-1}\ln(1-u)du \\
         & - \sum_{i=1}^n\frac{2\eta_2}{\varpi^2\left(1-e^{-(\varpi\ln y_{i})^2}\right)^{\eta_2}}\int_0^{1-e^{-(\varpi\ln y_{i})^2}}u^{\eta_2-1}\ln(1-u)du\\
         & + \sum_{i=1}^n\sum_{j=1}^v\frac{2\eta_1(\eta_1-1)}{\varpi^2\left(1-e^{-(\varpi\ln x_{ij})^2}\right)^{\eta_1}}\int_0^{1-e^{-(\varpi\ln x_{ij})^2}}u^{\eta_1-3}(1-u)\ln(1-u)(u+2\ln(1-u))du \\
         & + \sum_{i=1}^n\frac{2\eta_2(\eta_2-1)}{\varpi^2\left(1-e^{-(\varpi\ln y_{i})^2}\right)^{\eta_2}}\int_0^{1-e^{-(\varpi\ln y_{i})^2}}u^{\eta_2-3}(1-u)\ln(1-u)(u+2\ln(1-u))du. 
    \end{aligned}
\end{equation*}

\subsection*{Matrix elements for the Sect \ref{EFIM}}
\begin{equation*}
    \begin{aligned}
        e_{11}= & \frac{vn}{\eta_1^2},\quad e_{22}=\frac{n}{\eta_2^2},\quad e_{12}=0,\\
        e_{13}= & \frac{2\eta_1}{\varpi^2}\sum_{i=1}^n\sum_{j=1}^vC_{j-1}\sum_{l=1}^ja_{j,l}B(\eta_1\gamma_l,1)[\psi(1)-\psi(1+\eta_1\gamma_l)],\\
        e_{23}= & \frac{2\eta_2}{\varpi^2}\sum_{i=1}^nC_{i-1}\sum_{l=1}^ia_{i,l}B(\eta_2\gamma_l,1)[\psi(1)-\psi(1+\eta_2\gamma_l)], \\
        e_{33}= & \frac{2m(v+1)}{\varpi^2} -\frac{2\eta_1}{\varpi^2} \bigg[\sum_{i=1}^n\sum_{j=1}^v \left(\eta_1(1+R_j)-1 \right) C_{j-1}\sum_{l=1}^ja_{j,l} \bigg\{B(\eta_1\gamma_l-1,2)[\psi(2) \\
        & -\psi(\eta_1\gamma_l+1)] + B(\eta_1\gamma_l-2,2)[(\psi(2)-\psi(\eta_1\gamma_l))^2 + \psi^{'}(2)-\psi^{'}(\eta_1\gamma_l)]\bigg\}\bigg] \\
        & -\frac{2\eta_2}{\varpi^2} \bigg[\sum_{i=1}^n \left(\eta_2(1+S_i)-1 \right) C_{i-1}\sum_{l=1}^ia_{i,l} \bigg\{B(\eta_2\gamma_l-1,2)[\psi(2)-\psi(\eta_2\gamma_l+1)] \\
        & + B(\eta_2\gamma_l-2,2)[(\psi(2)-\psi(\eta_2\gamma_l))^2 + \psi^{'}(2)-\psi^{'}(\eta_2\gamma_l)]\bigg\}\bigg].
    \end{aligned}
\end{equation*}

\section*{Funding} The research work of Tanmay Kayal is partially supported by the Indian Institute of Technology Mandi, India, with the SEED project grant (Project No. IITM/SG/TK/141).

\bibliographystyle{plain}

\end{document}